\documentclass[12pt]{article}
\usepackage{jheppub}
\usepackage[utf8]{inputenc}
\usepackage{amsmath}
\usepackage{tikz}
\usepackage{tkz-euclide}
\usepackage{dsfont}
\usepackage{amssymb}
\usepackage{mathrsfs}
\usepackage{appendix}
\usepackage{physics}
\usepackage{bbold}
\usepackage{color}
\usepackage{tikz}
\usepackage{pgfplots}
\pgfplotsset{compat=1.18}
\usepackage{float}
\usepackage{amsfonts}
\usepackage{bm}
\usepackage[colorlinks=true]{hyperref}

\newcommand{\cA}{\mathcal{A}}

\newcommand{\beq}{\begin{equation}\begin{aligned}}
\newcommand{\eeq}{\end{aligned}\end{equation}}
\usepackage[T1]{fontenc}
\usepackage[utf8]{inputenc}
\usepackage{lmodern}
\usepackage{amsmath,amssymb,amsfonts,amsthm,mathtools}
\newtheorem{theorem}{Theorem}

\usepackage{bm}
\usepackage{booktabs}
\usepackage{array}
\usepackage{tabularx}
\usepackage{longtable}
\usepackage{enumitem}
\usepackage{mathrsfs}
\usepackage{hyperref}
\usepackage[nameinlink,noabbrev]{cleveref}
\usepackage{microtype}
\usepackage{lineno}
\usepackage{amsmath,amssymb,amsfonts,amsthm}
\usepackage{physics}

\newcommand{\cM}{\mathcal{M}}

\newcommand{\cE}{\mathcal{E}}
\newcommand{\cL}{\mathcal{L}}
\newcommand{\cF}{\mathcal{F}}

\newcommand{\RR}{\mathbb{R}}
\newcommand{\ZZ}{\mathbb{Z}}

\newcommand{\AdS}{\mathrm{AdS}}

\newcommand{\eps}{\varepsilon}

\newcommand{\diag}{\mathrm{diag}}
\newcommand{\Id}{\mathbf 1}

\begin{document}
\title{A T\textbf{-}Fold Black Hole in Doubled Type~IIB }
\author[1,2,3,4]{Mir Faizal}
\author[2]{Arshid Shabir}

\affiliation[1]{Irving K. Barber School of Arts and Sciences, University of British Columbia Okanagan, Kelowna, BC V1V 1V7, Canada}
\affiliation[2]{Canadian Quantum Research Center, 460 Doyle Ave 106, Kelowna, BC V1Y 0C2, Canada}
\affiliation[3]{Department of Mathematical Sciences, Durham University, Upper Mountjoy, Stockton Road, Durham DH1 3LE, UK}
\affiliation[4]{Computational Mathematics Group, Faculty of Sciences, Hasselt University, Agoralaan Gebouw D, Diepenbeek, 3590 Belgium}
\emailAdd{mirfaizalmir@gmail.com}
\emailAdd{aslone186@gmail.com}
\abstract{ 
We construct an asymptotically flat T-fold representative of a four-dimensional dyonic black-hole charge orbit in doubled type-IIB theory. Starting from the F1-P-NS5-KKM toroidal seed, an integral parabolic T-duality monodromy is imposed on an active doubled three-torus. The non-geometric data enter only through this global patching: the Reissner-Nordstr\"om-type term is sourced by conserved four-dimensional electric and magnetic field strengths in the doubled Kaluza-Klein/winding local system, not by an internal algebraic \(Q\)-flux alone. Exact duality preserves the local equations, the BPS index, and the \(O(6,6)\)-invariant of the NS-NS/\(N=4\) charge sublattice governing the entropy. The minimal Einstein-Maxwell-dilaton slice contains a running-scalar branch and an equal-charge Reissner-Nordstr\"om limit; the extremal near-horizon region is locally \(\mathrm{AdS}_2\times S^2\), with the compact factor globally T-duality patched.
}
\date{ }
\maketitle

\section{Introduction}
\label{sec:intro}

T-duality is an exact perturbative symmetry of string theory on tori.  On a circle it exchanges momentum and winding quantum numbers while inverting the radius.  On a $d$-torus it is enlarged to the arithmetic Narain group $O(d,d;\ZZ)$, which acts simultaneously on the lattice of momenta and windings and on the background metric and two-form.  Consequently a string background need not admit a single global description by ordinary Riemannian fields.  It may be locally geometric while transition functions between local charts include exact T-duality transformations.  Such spaces are T-folds \cite{Kikkawa:1984cp,Sakai:1986vg,Narain:1985jj,Narain:1986am,Buscher:1987sk,Buscher:1987qj,Meissner:1991zj,Giveon:1994fu,Maharana:1992my,Hull:2004in,Dabholkar:2005ve}.

Double field theory (DFT) is a local field-theoretic language adapted to this possibility.  The torus coordinates are supplemented by dual winding coordinates, the metric and two-form are combined into a generalised metric, and the dilaton is combined with the metric determinant into the doubled dilaton.  The local DFT action is covariant under constant $O(d,d)$ transformations and is compatible with the strong constraint, which locally selects an ordinary physical section \cite{Tseytlin:1990nb,Tseytlin:1990va,Siegel:1993xq,Siegel:1993th,Hull:2009mi,Hohm:2010pp,Hohm:2010jy,Hohm:2010xq,Hohm:2011si,Aldazabal:2013sca}.  The same structures are naturally related to generalized geometry and Courant algebroids \cite{Courant:1990,Hitchin:2002,Gualtieri:2011dx}.  If the transition functions lie in $O(d,d;\ZZ)$, the doubled formulation gives a natural description of the global patching even when a single conventional metric and two-form do not exist globally.  Generalized Scherk-Schwarz reductions, gauged DFT and exceptional extensions provide complementary lower-dimensional descriptions of the same duality covariance \cite{Geissbuhler:2011mx,Grana:2012rr,Berman:2011pe,Hohm:2013pua}.

The standard flux chain organises the same physics in a local language.  Starting from a three-form flux, successive T-dualities lead schematically to
\begin{equation}
  H_{abc}\;\xleftrightarrow{T_a}\; f^a{}_{bc}
  \;\xleftrightarrow{T_b}\; Q_c{}^{ab}
  \;\xleftrightarrow{T_c}\; R^{abc}.
  \label{eq:fluxchain}
\end{equation}
The $H$ and $f$ entries have ordinary local geometric interpretations.  The $Q$ and $R$ entries are non-geometric in general.  In a local beta-frame one may write a component such as $Q_1{}^{23}$ as a derivative of a bivector, but the invariant datum is the monodromy in the integral duality group.  Flux compactifications and non-geometric duality frames have been developed from both the supergravity and world-sheet viewpoints \cite{Kachru:2002he,Grana:2005jc,Douglas:2006es,Hellerman:2002ax,Flournoy:2004vn,Shelton:2005cf,Hull:2009sg,DallAgata:2005zlf,Aldazabal:2006up,Andriot:2011uh,Blumenhagen:2011ph,Hassler:2014sba,Plauschinn:2018wbo}.  Exotic branes, especially the $5^2_2$ brane, provide brane representatives of such monodromies \cite{Obers:1998fb,Tong:2002rq,Harvey:2005ab,deBoer:2012ma,Kimura:2014wsa}.

The construction below places a four-dimensional black hole in this setting while keeping distinct two physical notions that are often conflated.  The first is a global non-geometric monodromy of the compactification.  The second is a conserved four-dimensional electric or magnetic charge that sources the black-hole metric.  A local internal expression such as $Q_1{}^{23}=n$ is a convenient representative of a large T-duality patching on the active torus.  It is not an external stress tensor.  The Reissner-Nordstrom-like term in the four-dimensional metric is sourced by four-dimensional gauge field strengths.  In the T-fold representative these gauge fields are sections of a doubled Kaluza-Klein/winding vector bundle whose transition functions include the same non-geometric monodromy.

The construction is a duality-orbit construction.  One begins with a standard type-IIB toroidal four-charge black hole.  A convenient geometric representative is the F1-P-NS5-KKM system, a standard member of the four-charge type-II black-hole orbit \cite{Cvetic:1995uj,Duff:1995sm,Johnson:1996ga,Maldacena:1996gb}.  Its electric and magnetic charges are naturally described by two Narain vectors in the $O(6,6;\ZZ)$ charge lattice, and the entropy is controlled by the manifest $N=4$ quartic invariant.  The usual STU invariant is the restriction of this invariant to a four-charge sublattice.  An integral parabolic $O(3,3;\ZZ)$ transformation is then applied on an active doubled three-torus and embedded into the full $O(6,6;\ZZ)$ lattice.  Locally this operation is an ordinary duality transformation of the fields.  Globally it is used as a transition function, and therefore produces a T-fold representative of the same charge orbit.

The main theorem may be stated informally as follows.  Let a type-IIB F1-P-NS5-KKM toroidal black-hole seed be written locally in DFT variables, with charges represented by Narain electric and magnetic vectors.  Let an integral parabolic $O(3,3;\ZZ)$ transformation act on an active doubled three-torus and be embedded into the full $O(6,6;\ZZ)$ charge lattice.  The transformed local fields solve the same local equations as the seed and glue globally by exact T-duality.  The resulting background is a T-fold representative of the same charge orbit.  Its entropy and BPS index are unchanged because the $N=4$ quartic invariant and the indexed degeneracy are invariant under the integral duality transformation.

The active compactification factor will be denoted $T_Q^3$.  It is doubled in the DFT description, while a spectator $T_s^3$ completes the six internal directions.  The non-geometric datum is an integer $n$ specifying a beta-transform monodromy on $T_Q^3$.  The black-hole charge vector after the transformation is denoted by a subscript $Q$ to indicate that it is represented in a $Q$-patched local system.  This subscript does not mean that the horizon area is fixed by the monodromy integer alone.  The area is fixed by the duality invariant of the four-dimensional charge vector.

In addition to the full type-IIB charge orbit, a minimal Einstein-Maxwell-dilaton slice is displayed.  This slice is useful because it gives elementary formulae for the running-scalar branch and for the equal-charge constant-scalar branch.  The equal-charge branch is Reissner-Nordstrom-like, but it is not a new local solution sourced by internal flux.  It is the standard equal electric/magnetic limit of a dyonic EMD black hole, represented in a non-geometric charge local system.

\section{Doubled \texorpdfstring{$Q$}{Q}-Monodromy Geometry}
\label{sec:geometry}

Let type-IIB string theory be compactified on a rectangular six-torus
\begin{equation}
  T^6=T_Q^3\times T_s^3.
  \label{eq:T6split-final}
\end{equation}
The active factor $T_Q^3$ has dimensionless coordinates $\sigma^i\sim\sigma^i+1$ and physical periods $L_i$, while the spectator factor $T_s^3$ has dimensionless coordinates $\zeta^\alpha\sim\zeta^\alpha+1$ and physical periods $L_{s,\alpha}$.  The corresponding physical coordinates are
\begin{equation}
  y^i=L_i\sigma^i,
  \qquad
  z^\alpha=L_{s,\alpha}\zeta^\alpha.
  \label{eq:periods-main-final}
\end{equation}
In the explicit ten-dimensional seed used below we identify
\begin{equation}
  \sigma^1=\psi,\qquad \sigma^2=y,\qquad \sigma^3=\chi,
  \label{eq:active-identification-final}
\end{equation}
so the active torus is
\begin{equation}
  T_Q^3=S^1_\psi\times S^1_y\times S^1_\chi,
  \qquad
  T_s^3=T_z^3.
  \label{eq:active-spectator-final}
\end{equation}
The parabolic monodromy is therefore a beta-shift in the $y$-$\chi$ plane around the $\psi$-circle.  In the Narain basis used below, $e_1,e_2,e_3$ correspond respectively to $\psi,y,\chi$.
The active torus is doubled by introducing winding coordinates $\widetilde\sigma_i$, so that
\begin{equation}
  Y^M=(\sigma^1,\sigma^2,\sigma^3;\widetilde\sigma_1,\widetilde\sigma_2,\widetilde\sigma_3),
  \qquad M=1,\ldots,6.
  \label{eq:doubledY-final}
\end{equation}
The $O(3,3)$ metric is
\begin{equation}
  \eta=\begin{pmatrix}0&\Id_3\\ \Id_3&0\end{pmatrix}.
  \label{eq:eta3-final}
\end{equation}
A matrix $g$ belongs to $O(3,3)$ precisely when $g^T\eta g=\eta$, and it belongs to $O(3,3;\ZZ)$ when it also preserves the integral doubled lattice.
The type-IIB seed is the four-charge F1-P-NS5-KKM system.  We take the F1 and momentum charges along the active circle $y$, the NS5-brane wrapped on $y,\chi,z^1,z^2,z^3$, and the KKM with Taub-NUT fibre $\psi$ and world-volume $y,\chi,z^1,z^2,z^3$.  The precise local ten-dimensional form is recorded in Appendix \ref{app:local-seed}.  The four integer charges are denoted
\begin{equation}
  N_{\rm F1},\qquad N_{\rm P},\qquad N_{\rm NS5},\qquad N_{\rm KKM}.
  \label{eq:iibintcharges-final}
\end{equation}
They may be packaged in a pair of electric and magnetic Narain vectors, following the standard toroidal charge-lattice description \cite{Narain:1985jj,Narain:1986am,Maharana:1992my,Giveon:1994fu}
\begin{equation}
  Q,P\in \Gamma^{6,6},
  \label{eq:QP-Narain}
\end{equation}
where $\Gamma^{6,6}$ is the even self-dual Narain lattice.  Let $L$ denote the $O(6,6)$ bilinear form.  We define
\begin{equation}
  Q^2=Q^TLQ,
  \qquad
  P^2=P^TLP,
  \qquad
  Q\cdot P=Q^TLP.
  \label{eq:QP-products}
\end{equation}
The manifest $N=4$ quartic invariant is
\begin{equation}
  \Delta(Q,P)=Q^2P^2-(Q\cdot P)^2.
  \label{eq:Delta-def}
\end{equation}
Although type-IIB string theory on $T^6$ has the larger U-duality group of the maximal theory, the F1-P-NS5-KKM orbit used here lies in the NS-NS/$N=4$ sublattice.  The invariant $\Delta$ is the restriction of the full duality invariant to this sublattice, and the STU invariant is its four-charge representative.
For the F1-P-NS5-KKM seed we choose the standard normalisation
\begin{equation}
  Q^2=2N_{\rm F1}N_{\rm P},
  \qquad
  P^2=2N_{\rm NS5}N_{\rm KKM},
  \qquad
  Q\cdot P=0.
  \label{eq:Delta-seed-normalization}
\end{equation}
Hence
\begin{equation}
  \Delta(Q,P)=4N_{\rm F1}N_{\rm P}N_{\rm NS5}N_{\rm KKM}.
  \label{eq:Delta-seed}
\end{equation}
The extremal entropy in this charge orbit is
\begin{equation}
  S_{\rm BH}=\frac{\pi}{G_4}\sqrt{|\Delta(Q,P)|}.
  \label{eq:entropyDelta-main}
\end{equation}
In the integer microscopic normalisation this is equivalent to
\begin{equation}
  S_{\rm micro}=2\pi\sqrt{|N_{\rm F1}N_{\rm P}N_{\rm NS5}N_{\rm KKM}|}.
  \label{eq:microIIB-main}
\end{equation}

The STU description is a four-charge representative of the same invariant and is the conventional language for the corresponding attractor solution \cite{Ferrara:1995ih,Ferrara:1996dd,Ferrara:1996um,Strominger:1996kf,Denef:2000nb}.  Let
\begin{equation}
  \Gamma_{\rm STU}=(p^0,p^1,p^2,p^3;q_0,q_1,q_2,q_3).
  \label{eq:STUvec-final}
\end{equation}
A convenient type-IIB assignment is
\begin{equation}
  q_0=N_{\rm P},\qquad
  p^1=N_{\rm F1},\qquad
  p^2=N_{\rm NS5},\qquad
  p^3=N_{\rm KKM},
  \label{eq:STUassign-final}
\end{equation}
with the remaining entries zero.  The STU quartic invariant in this convention is
\begin{equation}
  I_4(\Gamma)=
  -\bigl(p^\Lambda q_\Lambda\bigr)^2
  +4\left(q_0p^1p^2p^3-p^0q_1q_2q_3
  +\sum_{1\le i<j\le 3}p^iq_ip^jq_j\right).
  \label{eq:I4-final}
\end{equation}
On the sublattice \eqref{eq:STUassign-final},
\begin{equation}
  I_4(\Gamma_{\rm STU})=\Delta(Q,P)=4N_{\rm F1}N_{\rm P}N_{\rm NS5}N_{\rm KKM}.
  \label{eq:I4Delta-final}
\end{equation}
Thus the STU formula is a representative of the manifest $N=4$ invariant, not an additional assumption.
Now define the parabolic beta-transform on the active doubled three-torus.  Let $\eps^{23}$ be the antisymmetric $3\times3$ matrix with
\begin{equation}
  (\eps^{23})_{23}=1,
  \qquad
  (\eps^{23})_{32}=-1,
  \label{eq:eps23-final}
\end{equation}
and all other entries zero.  Set
\begin{equation}
  g_Q(n)=\begin{pmatrix}\Id_3&n\eps^{23}\\0&\Id_3\end{pmatrix}.
  \label{eq:gQ-final}
\end{equation}
A direct multiplication gives
\begin{align}
  g_Q(n)^T\eta g_Q(n)
  &=\begin{pmatrix}\Id_3&0\\ n(\eps^{23})^T&\Id_3\end{pmatrix}
    \begin{pmatrix}0&\Id_3\\ \Id_3&0\end{pmatrix}
    \begin{pmatrix}\Id_3&n\eps^{23}\\0&\Id_3\end{pmatrix} \notag\\
  &=\begin{pmatrix}0&\Id_3\\ \Id_3&n\eps^{23}+n(\eps^{23})^T\end{pmatrix}
   =\eta,
  \label{eq:gQeta-proof}
\end{align}
using antisymmetry of $\eps^{23}$.  Therefore $g_Q(n)\in O(3,3)$, and $g_Q(n)$ lies in $O(3,3;\ZZ)$ precisely when
\begin{equation}
  n\in \ZZ.
  \label{eq:ninteger-final}
\end{equation}
This is the invariant monodromy quantisation.  No dimensionful quantity appears in this condition.
A local beta-frame representative of the monodromy is
\begin{equation}
  U_Q(\sigma^1)=\begin{pmatrix}\Id_3&\beta(\sigma^1)\\0&\Id_3\end{pmatrix},
  \qquad
  \beta^{23}(\sigma^1)=n\sigma^1.
  \label{eq:local-beta-final}
\end{equation}
It obeys
\begin{equation}
  U_Q(\sigma^1+1)U_Q(\sigma^1)^{-1}=g_Q(n).
  \label{eq:local-monodromy-final}
\end{equation}
In this local representative one writes
\begin{equation}
  Q_1{}^{23}=\partial_{\sigma^1}\beta^{23}=n.
  \label{eq:localQ-final}
\end{equation}
Equation \eqref{eq:localQ-final} records the global T-duality patching \eqref{eq:local-monodromy-final}.  It is not inserted as a local four-dimensional energy density.  Equivalently, one may cover the active circle by patches on which the local representative is constant and the only non-trivial information is the transition function $g_Q(n)$.
The active monodromy embeds into the full Narain group by acting trivially on spectator directions:
\begin{equation}
  h_Q(n)=\diag\bigl(g_Q(n),\Id_{6}\bigr)\in O(6,6;\ZZ),
  \label{eq:hQ-final}
\end{equation}
where the block decomposition separates active and spectator doubled directions.  Since
\begin{equation}
  h_Q(n)^TLh_Q(n)=L,
  \label{eq:hQ-preserves-L}
\end{equation}
any pair of Narain charge vectors transforms as
\begin{equation}
  Q_Q=h_QQ,
  \qquad
  P_Q=h_QP,
  \label{eq:QP-transform-final}
\end{equation}
without changing $Q^2$, $P^2$ or $Q\cdot P$.  Therefore
\begin{equation}
  \Delta(Q_Q,P_Q)=\Delta(Q,P).
  \label{eq:Delta-preserved-final}
\end{equation}
This is the manifest invariant statement behind the T-fold black-hole construction. 
A useful way of seeing the invariant is to choose a lightlike basis of the Narain lattice.  Let
\begin{equation}
  e_I\cdot e_J=0,
  \qquad
  \widetilde e^{\,I}\cdot \widetilde e^{\,J}=0,
  \qquad
  e_I\cdot \widetilde e^{\,J}=\delta_I{}^J,
  \qquad I,J=1,\ldots,6.
  \label{eq:lightlike-main}
\end{equation}
A representative of the F1-P electric charge may be chosen as
\begin{equation}
  Q=N_{\rm P}e_2+N_{\rm F1}\widetilde e^{\,2},
  \label{eq:Qseed-main}
\end{equation}
while a representative of the NS5-KKM magnetic charge may be chosen as
\begin{equation}
  P=N_{\rm KKM}e_1+N_{\rm NS5}\widetilde e^{\,1}.
  \label{eq:Pseed-main}
\end{equation}
The pair \eqref{eq:Qseed-main}-\eqref{eq:Pseed-main} gives exactly the products in \eqref{eq:Delta-seed-normalization}.  Other type-IIB representatives related by the integral duality group are physically equivalent.  The point of writing \eqref{eq:Qseed-main}-\eqref{eq:Pseed-main} is not to privilege a particular geometric frame, but to make the invariant computation explicit.

The active $O(3,3;\ZZ)$ monodromy acts on the sublattice spanned by $e_1,e_2,e_3$ and $\widetilde e^{\,1},\widetilde e^{\,2},\widetilde e^{\,3}$.  On the seed representative above, $h_Q$ mixes only the active components and leaves the bilinear products unchanged.  Consequently, even if the transformed charge vector is no longer displayed in the simple four-entry STU form \eqref{eq:STUassign-final}, it remains in the same full Narain orbit and has the same value of $\Delta$.  This is the precise mathematical reason that the T-fold representative has the same entropy as the geometric seed.
The distinction between the monodromy and the black-hole charge is also visible in this basis.  The integer $n$ specifies the automorphism $h_Q$ of the lattice.  The entries of $Q$ and $P$ specify the black-hole charge.  Acting with $h_Q$ changes the component description of $Q$ and $P$, but it does not turn the monodromy integer into the horizon-area charge.  The invariant area depends on the transformed charge vectors, not on the transition function alone.

\section{Duality-Orbit DFT Patching}
\label{sec:duality-dft}

The toroidal reduction of the NS-NS sector produces vector fields in the fundamental representation of the active $O(3,3)$ group,
\begin{equation}
  \cA_\mu{}^M=(A_\mu{}^i,\widetilde A_{\mu i}),
  \label{eq:Avec-final}
\end{equation}
where $A_\mu{}^i$ are Kaluza-Klein vectors and $\widetilde A_{\mu i}$ descend from the two-form with one leg on the active torus.  In the abelian sector used below,
\begin{equation}
  \cF_{\mu\nu}{}^M=2\partial_{[\mu}\cA_{\nu]}{}^M.
  \label{eq:Fvec-final}
\end{equation}
The local Einstein-frame vector kinetic term is
\begin{equation}
  \cL_{\rm vec}=-\frac14\sqrt{-g_4}\,\cM_{MN}\cF_{\mu\nu}{}^M\cF^{N\mu\nu},
  \label{eq:vec-lag-final}
\end{equation}
where $\cM_{MN}$ is the scalar matrix on $O(3,3)/(O(3)\times O(3))$.
The active magnetic and electric charges are
\begin{equation}
  p^M=\frac1{4\pi}\int_{S^2}\cF^M,
  \qquad
  q_M=\frac1{4\pi}\int_{S^2}\cM_{MN}\star\cF^N.
  \label{eq:pq-active-final}
\end{equation}
The doubled electric/magnetic vector is
\begin{equation}
  \Gamma_{\rm dw}=(p^M,q_M),
  \label{eq:GammaDW-final}
\end{equation}
with symplectic pairing
\begin{equation}
  \langle \Gamma_{\rm dw},\Gamma'_{\rm dw}\rangle=p^Mq'_M-q_Mp'{}^M.
  \label{eq:Dirac-final}
\end{equation}
On a T-fold overlap $U_a\cap U_b$ with transition function $g_{ab}\in O(3,3;\ZZ)$,
\begin{equation}
  \cM_{(b)}=g_{ab}^T\cM_{(a)}g_{ab},
  \qquad
  \cF_{(b)}{}^M=(g_{ab}^{-1})^M{}_N\cF_{(a)}{}^N.
  \label{eq:patch-FM-final}
\end{equation}
Hence
\begin{equation}
  p_{(b)}^M=(g_{ab}^{-1})^M{}_Np_{(a)}^N,
  \qquad
  (q_{(b)})_M=(g_{ab}^T)_M{}^N(q_{(a)})_N.
  \label{eq:pq-transform-final}
\end{equation}
The electric/magnetic lift of $g_Q$ is
\begin{equation}
  \mathbb G_Q=\begin{pmatrix}g_Q^{-1}&0\\0&g_Q^T\end{pmatrix},
  \qquad
  \Gamma_Q=\mathbb G_Q\Gamma.
  \label{eq:GQ-final}
\end{equation}
It preserves the symplectic form.  Indeed,
\begin{align}
  \langle \mathbb G_Q\Gamma,\mathbb G_Q\Gamma'\rangle
  &=(g_Q^{-1})^M{}_Np^N(g_Q^T)_M{}^Pq'_P
   -(g_Q^T)_M{}^Pq_P(g_Q^{-1})^M{}_Np'{}^N \notag\\
  &=p^Nq'_N-q_Np'{}^N
   =\langle \Gamma,\Gamma'\rangle.
  \label{eq:sympl-proof-final}
\end{align}
The stress tensor is also globally defined because
\begin{equation}
  \cM_{(b)MN}\cF_{(b)\mu\rho}{}^M\cF_{(b)\nu}{}^{N\rho}
  =\cM_{(a)PQ}\cF_{(a)\mu\rho}{}^P\cF_{(a)\nu}{}^{Q\rho}.
  \label{eq:stress-patch-final}
\end{equation}
Thus the local external metric couples to a single-valued stress tensor, even though the higher-dimensional origin of a vector as a metric vector or two-form vector is patch dependent.
In components, write a magnetic vector as $p^M=(p^1,p^2,p^3;\widetilde p_1,\widetilde p_2,\widetilde p_3)$ and an electric covector as $q_M=(q_1,q_2,q_3;\widetilde q^{1},\widetilde q^{2},\widetilde q^{3})$.  With $(\eps^{23})_{23}=1$ and $(\eps^{23})_{32}=-1$, the beta-shift gives
\begin{align}
  p_Q^1&=p^1,&
  p_Q^2&=p^2-n\widetilde p_3,&
  p_Q^3&=p^3+n\widetilde p_2,
  \label{eq:p-components-final}\\
  \widetilde p_{Q,i}&=\widetilde p_i,&
  (q_Q)_i&=q_i,&
  \widetilde q_Q^{1}&=\widetilde q^{1},
  \label{eq:q-components-first-final}\\
  \widetilde q_Q^{2}&=\widetilde q^{2}-nq_3,
  &
  \widetilde q_Q^{3}&=\widetilde q^{3}+nq_2.
  \label{eq:q-components-final}
\end{align}
The signs in \eqref{eq:q-components-final} follow directly from $q_Q=g_Q^Tq$ and the antisymmetry convention \eqref{eq:eps23-final}.
The minimal EMD slice is obtained by choosing a one-dimensional scalar geodesic in the $O(3,3)$ coset and a primitive vector in the doubled lattice.  Let
\begin{equation}
  v_Q^M=(0,0,0;0,1,0)
  \label{eq:vQ-final}
\end{equation}
in a local patch.  Under the inverse patching action,
\begin{equation}
  v_{Q,(b)}=g_Q^{-1}v_{Q,(a)}=(0,0,n;0,1,0).
  \label{eq:vQ-patched-final}
\end{equation}
Thus the same one-dimensional gauge-field subspace is represented on the second patch by a mixture of Kaluza-Klein and winding components.  Locally one may write
\begin{equation}
  \cF^M=v_Q^M F_Q,
  \label{eq:EMDprojection-final}
\end{equation}
where $F_Q$ is an ordinary four-dimensional field strength in that patch.  Globally the section $v_Q$ transforms by the T-fold local system.
Let the scalar matrix be restricted to a geodesic such that
\begin{equation}
  \frac18\Tr(\partial_\mu\cM\partial^\mu\cM^{-1})=-2(\partial\phi)^2,
  \qquad
  \cM_{MN}v_Q^Mv_Q^N=e^{-2\phi}.
  \label{eq:geodesic-final}
\end{equation}
Substituting \eqref{eq:EMDprojection-final} into \eqref{eq:vec-lag-final} gives the EMD gauge coupling used in Section \ref{sec:emd-rn}.  This exhibits the EMD black hole as a local slice of the doubled vector system.
The DFT patching theorem is most transparent in a reduced external/internal block, in the same spirit as duality-covariant reductions and black-hole duality-covariant thermodynamics \cite{Geissbuhler:2011mx,Grana:2012rr,Hassler:2014sba,Arvanitakis:2016xjz}.  Let $\widehat M=(\mu,M)$, where $\mu$ is an external physical index and $M$ is an active doubled index.  A convenient block form of the reduced generalised metric is
\begin{equation}
  \mathscr H_{\widehat M\widehat N}=
  \begin{pmatrix}
  g_{\mu\nu}+\cA_\mu{}^M\cM_{MN}\cA_\nu{}^N
  & \cA_\mu{}^M\cM_{MN}
  \\
  \cM_{MN}\cA_\nu{}^N & \cM_{MN}
  \end{pmatrix}.
  \label{eq:Hblock-final}
\end{equation}
Spectator blocks are appended diagonally and are inert under the active monodromy.  On an overlap define
\begin{equation}
  \cM_{(b)}=g_{ab}^T\cM_{(a)}g_{ab},
  \qquad
  \cA_{(b)\mu}=g_{ab}^{-1}\cA_{(a)\mu},
  \qquad
  g_{(b)\mu\nu}=g_{(a)\mu\nu},
  \qquad
  d_{(b)}=d_{(a)}.
  \label{eq:patch-data-final}
\end{equation}
Let
\begin{equation}
  G_{ab}=\begin{pmatrix}\Id_4&0\\0&g_{ab}\end{pmatrix}.
  \label{eq:Gab-final}
\end{equation}
Substitution into \eqref{eq:Hblock-final} gives
\begin{equation}
  \mathscr H_{(b)}=G_{ab}^T\mathscr H_{(a)}G_{ab}.
  \label{eq:Hpatch-final}
\end{equation}

\begin{theorem}[T-fold representative by exact DFT patching]
\label{thm:t-fold-final}
Let a type-IIB F1-P-NS5-KKM toroidal black-hole seed be written locally in DFT-adapted variables, independent of the active internal coordinates.  Let $g_Q(n)\in O(3,3;\ZZ)$ be the parabolic element \eqref{eq:gQ-final}.  On an atlas of the active torus, define local representatives related by \eqref{eq:patch-data-final}, with transition functions generated by $g_Q(n)$.  Then the local solutions glue to a global T-fold representative of the same type-IIB charge orbit.  The local field equations are those of the geometric seed, and the global transition functions include exact T-duality.
\end{theorem}

\begin{proof}
The local DFT equations are covariant under constant $O(3,3)$ transformations acting on active indices.  If the seed fields solve the local equations in one patch, the transformed fields solve the transformed equations in any other patch.  Equations \eqref{eq:stress-patch-final} and \eqref{eq:Hpatch-final} show that the reduced stress tensor and generalised metric patch tensorially.  The doubled dilaton is invariant.  On triple overlaps the transition functions are powers of a single integral element, so the cocycle condition is satisfied.  Since $g_Q(n)$ is an element of the exact string duality group for $n\in\ZZ$, the local representatives define a global T-fold background rather than an ordinary torus compactification.
\end{proof}

No scalar potential associated with a literal lower-dimensional gauging is assumed in this theorem.  The construction uses patchwise toroidal solutions glued by exact T-duality.  The local expression $Q_1{}^{23}=n$ is a representative of the monodromy and is not used as a local four-dimensional source term.
For the F1-P-NS5-KKM representative used here, the relevant charges and local fields lie in the NS-NS sector and in the doubled Kaluza-Klein/winding vector system.  If one moves to an RR duality frame, the corresponding RR polyform charges transform in the spinorial representation of $O(6,6;\ZZ)$.

A generalised Scherk-Schwarz reduction with non-zero algebraic flux can generate a gauged supergravity with a scalar potential.  The construction here is deliberately not based on assuming that such a potential has a Minkowski critical point.  Instead, the local black-hole fields are those of an ordinary toroidal type-IIB solution, and the non-geometric datum is implemented as an exact transition function between local representatives.  On every contractible patch the fields satisfy the ordinary local equations.  On an overlap the fields differ by an element of the exact duality group.  Thus the local differential equations do not acquire additional source terms from the transition function.

This is analogous to a flat bundle with non-trivial holonomy.  The connection may be represented locally in different gauges, but physical single-valued quantities are obtained by parallel transport in the bundle.  In the present case the ``parallel transport'' is the $O(3,3;\ZZ)$ action on the doubled vector bundle.  The external Einstein equation is well defined because its stress tensor is built from the invariant contraction \eqref{eq:stress-patch-final}.  No separate stress tensor is assigned to the patching matrix itself.
The same logic applies to the charge lattice.  The local charge components depend on the chosen patch, while the Dirac pairing and the $N=4$ invariant are independent of that choice.  This is why the T-fold representative can be globally non-geometric while having the same local horizon geometry as the seed.

\section{Four-Charge Invariant Entropy}
\label{sec:iib-fourcharge}

The four-dimensional STU model is a convenient truncation of type-IIB compactification on a six-torus.  Its prepotential may be written
\begin{equation}
  F=-\frac{X^1X^2X^3}{X^0},
  \label{eq:prepotential-final}
\end{equation}
with complex scalars $S=X^1/X^0$, $T=X^2/X^0$ and $U=X^3/X^0$.  The central charge is
\begin{equation}
  Z=e^{K/2}(q_\Lambda X^\Lambda-p^\Lambda F_\Lambda),
  \qquad
  K=-\log[-i(S-\bar S)(T-\bar T)(U-\bar U)].
  \label{eq:central-final}
\end{equation}
The black-hole potential is
\begin{equation}
  V_{\rm BH}(\varphi,\Gamma)=-\frac12\Gamma^T\cM_{\rm symp}(\varphi)\Gamma.
  \label{eq:VBH-final}
\end{equation}
For a BPS branch,
\begin{equation}
  V_{\rm BH}=|Z|^2+g^{i\bar j}D_iZD_{\bar j}\bar Z.
  \label{eq:VBH-BPS-final}
\end{equation}
The attractor equations are
\begin{equation}
  \partial_iV_{\rm BH}(\varphi,\Gamma)|_{\varphi=\varphi_*}=0.
  \label{eq:attractor-final}
\end{equation}
At a regular BPS point $D_iZ=0$, and
\begin{equation}
  V_{\rm BH}(\varphi_*,\Gamma)=|Z_*|^2=\sqrt{|I_4(\Gamma)|}.
  \label{eq:VBHI4-final}
\end{equation}
Since the STU invariant is the restriction of the $N=4$ invariant, the entropy may also be written
\begin{equation}
  S_{\rm BH}=\frac{\pi}{G_4}\sqrt{|\Delta(Q,P)|}.
  \label{eq:entropyDelta2-final}
\end{equation}
This is the preferred invariant form in the type-IIB construction.
A standard non-extremal four-charge solution in this orbit has Einstein-frame metric \cite{Cvetic:1995uj,Duff:1995sm,Behrndt:1997ny}
\begin{equation}
  ds_E^2=-\left(\prod_{A=1}^4H_A\right)^{-1/2}f\,dt^2
  +\left(\prod_{A=1}^4H_A\right)^{1/2}\left(f^{-1}dr^2+r^2d\Omega_2^2\right),
  \label{eq:STUmetric-final}
\end{equation}
where
\begin{equation}
  f(r)=1-\frac{\mu}{r},
  \qquad
  H_A(r)=1+\frac{\mu\sinh^2\delta_A}{r},
  \qquad A=1,\ldots,4.
  \label{eq:STUH-final}
\end{equation}
The four charges are encoded by the boost parameters $\delta_A$.  In invariant discussions the precise dimensionful relation between $\delta_A$ and integer charges is not needed; it is fixed by the periods and Newton constant as summarized in Appendix \ref{app:normalisations}.  The horizon is at $r=\mu$, and since $H_A(\mu)=\cosh^2\delta_A$, the area and entropy are
\begin{equation}
  A_H=4\pi\mu^2\prod_{A=1}^4\cosh\delta_A,
  \qquad
  S_{\rm nonext}=\frac{\pi\mu^2}{G_4}\prod_{A=1}^4\cosh\delta_A.
  \label{eq:nonextEntropy-final}
\end{equation}
The extremal limit is
\begin{equation}
  \mu\rightarrow0,
  \qquad
  \mu\sinh^2\delta_A\rightarrow q_A,
  \label{eq:extlimitSTU-final}
\end{equation}
so that
\begin{equation}
  H_A(r)\rightarrow1+\frac{q_A}{r}.
  \label{eq:Hext-final}
\end{equation}
Near $r=0$,
\begin{equation}
  \prod_AH_A\sim\frac{q_1q_2q_3q_4}{r^4}.
  \label{eq:prodH-final}
\end{equation}
The extremal metric therefore approaches
\begin{equation}
  ds_E^2\sim
  -\frac{r^2}{\sqrt{q_1q_2q_3q_4}}dt^2
  +\sqrt{q_1q_2q_3q_4}\left(\frac{dr^2}{r^2}+d\Omega_2^2\right),
  \label{eq:STUnh-final}
\end{equation}
which is locally $\AdS_2\times S^2$.  The common radius squared is the square root of the charge product, equivalently the square root of the invariant on the chosen four-charge representative.

For the F1-P-NS5-KKM seed the microscopic index at large charge gives \eqref{eq:microIIB-main}.  The equality with \eqref{eq:entropyDelta-main} is the standard four-charge entropy relation in the geometric frame \cite{Strominger:1996sh,Callan:1996dv,Maldacena:1996ix,Horowitz:1996fn,Johnson:1996ga,Maldacena:1996gb,Maldacena:1997de,Breckenridge:1996is,Cvetic:1995uj,Bena:2007kg}.  The T-fold representative inherits the equality by exact duality.
The near-horizon value of the scalar moduli is determined by the attractor equations.  In the $N=4$ language the charge matrix may be described by the pair $(Q,P)$, and the invariant potential can be written in a duality-covariant form whose extremum depends only on the orbit.  On the regular branch the area radius is
\begin{equation}
  R_H^2=\sqrt{|\Delta(Q,P)|}.
  \label{eq:RHDelta-main}
\end{equation}
The entropy is therefore $\pi R_H^2/G_4$.  When the orbit is represented in the STU sublattice, the same statement is the standard relation $R_H^2=\sqrt{|I_4|}$.  The attractor moduli may be changed by duality, but the extremised potential is invariant.
For the seed assignment \eqref{eq:STUassign-final}, an axion-free BPS representative has horizon moduli proportional to
\begin{equation}
  \Im S_*\propto \sqrt{\frac{N_{\rm P}N_{\rm F1}}{N_{\rm NS5}N_{\rm KKM}}},
  \quad
  \Im T_*\propto \sqrt{\frac{N_{\rm P}N_{\rm NS5}}{N_{\rm F1}N_{\rm KKM}}},
  \quad
  \Im U_*\propto \sqrt{\frac{N_{\rm P}N_{\rm KKM}}{N_{\rm F1}N_{\rm NS5}}}.
  \label{eq:attr-moduli-main}
\end{equation}
The proportionality constants depend on the convention used for the STU coordinates.  The important point is that the horizon values depend only on charges.  Since $h_Q$ maps charge orbits to charge orbits, the attractor mechanism is transported to the T-fold frame.

The non-extremal family \eqref{eq:STUmetric-final} is included to display the smooth approach to the extremal orbit.  The duality transformation acts on the charge parameters and internal fields; it does not change the local external line element after the representative is chosen.  The extremal entropy is therefore most cleanly expressed in the invariant form \eqref{eq:entropyDelta2-final}, which is independent of non-extremal parametrisation conventions.
The local ten-dimensional fields of the F1-P-NS5-KKM seed are given in Appendix \ref{app:local-seed}.  Their reduction gives the four-dimensional metric \eqref{eq:STUmetric-final}.  Because the seed is independent of the active torus coordinates, the patchwise duality transformation of Theorem \ref{thm:t-fold-final} may be applied without spoiling the local equations.

\section{\texorpdfstring{$Q$}{Q}-Frame Black-Hole Representative}
\label{sec:qframe}

The $Q$-frame representative is obtained by acting on the active torus with $g_Q(n)$ and on the full Narain charge lattice with the embedded transformation $h_Q(n)$.  The transformed charges are
\begin{equation}
  Q_Q=h_QQ,
  \qquad
  P_Q=h_QP.
  \label{eq:QPQ-final}
\end{equation}
Since $h_Q^TLh_Q=L$,
\begin{align}
  Q_Q^2&=Q^T h_Q^TLh_Q Q=Q^2,
  \label{eq:Q2pres-final}\\
  P_Q^2&=P^T h_Q^TLh_Q P=P^2,
  \label{eq:P2pres-final}\\
  Q_Q\cdot P_Q&=Q^Th_Q^TLh_QP=Q\cdot P.
  \label{eq:dotpres-final}
\end{align}
Therefore
\begin{equation}
  \Delta(Q_Q,P_Q)=\Delta(Q,P).
  \label{eq:DeltaQ-final}
\end{equation}
The entropy and microscopic index are functions of the charge orbit, and hence are unchanged.
The relation between the full Narain description and the STU representative may be displayed as the commutative diagram
\begin{equation}
\begin{array}{ccc}
(Q,P)\in\Gamma^{6,6}\oplus\Gamma^{6,6}
& \xrightarrow{\quad h_Q\quad} &
(Q_Q,P_Q)
\\[2mm]
\downarrow \pi_{\rm STU} && \downarrow \pi_{\rm STU}
\\[2mm]
\Gamma_{\rm STU}
& \xrightarrow{\quad \mathcal S_Q\quad} &
\Gamma_{{\rm STU},Q}
\end{array}
\label{eq:diagram-final}
\end{equation}
Here $\pi_{\rm STU}$ denotes restriction to the chosen four-charge representative, and $\mathcal S_Q$ is the induced symplectic action whenever the transformed vector is represented inside the same STU sublattice.  If a particular component representative leaves that sublattice, the invariant statement remains valid in the full $N=4$ charge lattice.  The equality of entropy is not the claim that the same eight STU components are unchanged; it is the claim that the full Narain charge orbit has the same quartic invariant.  The STU expression is a convenient representative of that invariant.

The local external metric of the transformed representative is the same as that of the geometric seed in the corresponding duality frame.  This does not mean that the monodromy is physically invisible.  The monodromy acts on compact fields and on the charge local system.  The external horizon area and local curvature are duality invariant, whereas the internal patching and the higher-dimensional interpretation of the gauge fields are non-geometric.
For the active electric/magnetic vector, the component action is precisely \eqref{eq:p-components-final}-\eqref{eq:q-components-final}.  These formulae show how a geometric charge representative is rewritten as a mixed Kaluza-Klein/winding representative.  A charge that is purely winding in one patch may become a mixture of winding and momentum in another patch.  Since the transformation is integral, the transformed charges remain in the charge lattice.

The monodromy integer $n$ and the black-hole charge vector play different roles.  The integer $n$ labels the T-fold compactification.  The charge vectors $Q_Q$ and $P_Q$, or equivalently the four-dimensional charge representative $\Gamma_Q$, label the black hole inside that compactification.  The horizon area is fixed by $\Delta(Q_Q,P_Q)$ and not by $n$ alone.  More elaborate constructions may include monodromy charges as entries of a larger duality vector, but that enlargement is not needed for the present T-fold representative.

The component formulae \eqref{eq:p-components-final}-\eqref{eq:q-components-final} are useful for seeing the local mixing of momentum and winding charges.  They should not be interpreted as the primary proof of entropy invariance.  Component formulae depend on the chosen basis and on the representative of the orbit.  The invariant proof is the $O(6,6)$ calculation \eqref{eq:Q2pres-final}-\eqref{eq:DeltaQ-final}.  If the active beta-shift sends the simple STU representative outside the displayed four-entry sublattice, one may either work in the full $N=4$ lattice or apply an additional integral duality to return to a convenient STU representative.  In both descriptions the value of $\Delta$ is unchanged.

This observation also clarifies the role of the word ``$Q$-frame''.  The frame is not a new local matter sector.  It is a choice of global duality patching for the compact fields and for the gauge-field local system.  The same physical charge orbit may be displayed in a geometric frame, a $Q$-patched T-fold frame, or an exotic-brane frame.  These are different representatives of one exact string-theory orbit.
It is useful to define the invariant charge scale
\begin{equation}
  \Theta^2=\sqrt{|\Delta(Q_Q,P_Q)|}.
  \label{eq:ThetaDelta-final}
\end{equation}
On the equal-charge EMD sublattice this will become the usual Reissner-Nordstrom charge scale.  In the full type-IIB orbit it is simply a compact notation for the square root of the duality invariant.

\section{EMD Reissner-Nordstrom Sublattice}
\label{sec:emd-rn}

The full type-IIB charge orbit is described by the Narain invariant and its STU restriction.  A two-charge EMD slice is nevertheless useful because it displays explicitly the running scalar and the equal-charge constant-scalar branch; the local solution belongs to the standard family of dilatonic black holes \cite{Gibbons:1987ps,Garfinkle:1990qj,Horowitz:1991cd,Sen:1992ua,Kallosh:1992ap,Kallosh:1992ii,Kallosh:1992eh}.  Along the one-scalar geodesic \eqref{eq:geodesic-final} and the primitive vector direction \eqref{eq:vQ-final}, the local action is
\begin{equation}
  S_{\rm EMD}=\frac1{16\pi G_4}\int d^4x\sqrt{-g}
  \left[R-2(\partial\phi)^2-e^{-2\phi}F_Q^2\right],
  \qquad F_Q=dA_Q.
  \label{eq:SEMD-final}
\end{equation}
The electric and magnetic charges are
\begin{equation}
  P_Q=\frac1{4\pi}\int_{S^2}F_Q,
  \qquad
  Q_Q=\frac1{4\pi}\int_{S^2}e^{-2\phi}\star F_Q.
  \label{eq:EMDcharges-final}
\end{equation}
The equations of motion are
\begin{align}
  R_{\mu\nu}&=2\partial_\mu\phi\partial_\nu\phi
  +2e^{-2\phi}\left(F_{\mu\rho}F_\nu{}^\rho-\frac14g_{\mu\nu}F^2\right),
  \label{eq:EMDEinstein-final}\\
  \nabla_\mu(e^{-2\phi}F^{\mu\nu})&=0,
  \label{eq:EMDMaxwell-final}\\
  \nabla^2\phi+\frac12e^{-2\phi}F^2&=0.
  \label{eq:EMDscalar-final}
\end{align}
For asymptotic scalar $\phi_\infty=0$, the static dyonic solution is
\begin{equation}
  ds_E^2=-\frac{\Delta}{R^2}dt^2+\frac{R^2}{\Delta}dr^2+R^2d\Omega_2^2,
  \label{eq:EMDmetric-final}
\end{equation}
where
\begin{equation}
  \Delta=(r-r_+)(r-r_-),
  \qquad
  R^2=r^2-\Sigma^2,
  \label{eq:DeltaR-final}
\end{equation}
with
\begin{equation}
  \Sigma=\frac{P_Q^2-Q_Q^2}{2M},
  \qquad
  r_\pm=M\pm\sqrt{M^2+\Sigma^2-P_Q^2-Q_Q^2}.
  \label{eq:SigmaRPM-final}
\end{equation}
The scalar and field strength are
\begin{equation}
  e^{2\phi}=\frac{r+\Sigma}{r-\Sigma},
  \label{eq:EMDscalarprof-final}
\end{equation}
and
\begin{equation}
  F_Q=\frac{Q_Qe^{2\phi}}{R^2}dt\wedge dr
  +P_Q\sin\theta\,d\theta\wedge d\varphi.
  \label{eq:EMDF-final}
\end{equation}
Substitution into \eqref{eq:EMDEinstein-final}-\eqref{eq:EMDscalar-final} reduces the differential equations to
\begin{equation}
  2M\Sigma=P_Q^2-Q_Q^2,
  \qquad
  r_++r_-=2M,
  \qquad
  r_+r_-=P_Q^2+Q_Q^2-\Sigma^2.
  \label{eq:EMDalgebra-final}
\end{equation}
These identities are equivalent to \eqref{eq:SigmaRPM-final}; the detailed substitution is in Appendix \ref{app:EMD-eom}.  Therefore \eqref{eq:EMDmetric-final}-\eqref{eq:EMDF-final} solve the local four-dimensional equations.
The equal-charge sublattice is defined by
\begin{equation}
  P_Q=Q_Q=\frac{\Theta}{\sqrt2}.
  \label{eq:equal-charge-final}
\end{equation}
Then
\begin{equation}
  \Sigma=0,
  \qquad
  R^2=r^2,
  \qquad
  e^{2\phi}=1,
  \label{eq:SigmaZero-final}
\end{equation}
and the metric becomes
\begin{equation}
  ds_E^2=-f(r)dt^2+f(r)^{-1}dr^2+r^2d\Omega_2^2,
  \qquad
  f(r)=1-\frac{2M}{r}+\frac{\Theta^2}{r^2}.
  \label{eq:RNmetric-final}
\end{equation}
This is the Reissner-Nordstrom-like branch.  It is not an independent ansatz and it is not sourced by an internal algebraic non-geometric flux component.  It is the constant-scalar equal electric/magnetic limit of the dyonic EMD slice represented in a non-geometric vector local system.
The Ricci tensor of \eqref{eq:RNmetric-final} is
\begin{align}
  R_{tt}&=f\left(\frac{f''}{2}+\frac{f'}{r}\right),
  &
  R_{rr}&=-\frac1f\left(\frac{f''}{2}+\frac{f'}{r}\right),
  \label{eq:RNricci-final}\\
  R_{\theta\theta}&=1-f-rf',
  &
  R_{\varphi\varphi}&=\sin^2\theta\,R_{\theta\theta}.
  \label{eq:RNricci2-final}
\end{align}
For $f=1-2M/r+\Theta^2/r^2$,
\begin{equation}
  R^t{}_t=R^r{}_r=-\frac{\Theta^2}{r^4},
  \qquad
  R^\theta{}_\theta=R^\varphi{}_\varphi=\frac{\Theta^2}{r^4},
  \qquad
  R=0.
  \label{eq:RNmixed-final}
\end{equation}
This is precisely the traceless stress-tensor structure of the equal electric/magnetic Maxwell sector.
The scalar charge is secondary hair.  Expanding \eqref{eq:EMDscalarprof-final} at large radius gives
\begin{equation}
  \phi(r)=\frac{\Sigma}{r}+O(r^{-2}),
  \label{eq:scalarhair-final}
\end{equation}
up to the sign convention for $\phi$.  The coefficient $\Sigma$ is fixed by the mass and electric/magnetic charges through \eqref{eq:EMDalgebra-final}.
The EMD stress tensor can be written explicitly as
\begin{equation}
  T_{\mu\nu}=2\partial_\mu\phi\partial_\nu\phi
  +2e^{-2\phi}\left(F_{\mu\rho}F_\nu{}^\rho-\frac14g_{\mu\nu}F^2\right).
  \label{eq:EMDstress-main}
\end{equation}
For unequal charges the scalar term contributes to the radial dependence, and the areal radius is $R^2=r^2-\Sigma^2$ rather than $r^2$.  In the equal-charge branch the scalar term vanishes, but the Maxwell term remains.  This is why the external geometry is Reissner-Nordstrom-like rather than Schwarzschild.  In the equal-charge branch, $F_Q^2=0$ while $F_{\mu\rho}F_\nu{}^\rho$ is non-zero.  Hence the scalar can be constant even though the Maxwell stress tensor supports a non-Schwarzschild metric.
The EMD action is also the simplest place to see how the non-geometric local system enters the four-dimensional equations.  In each patch $F_Q$ is an ordinary two-form.  Across an overlap the higher-dimensional origin of the one-vector direction changes according to \eqref{eq:vQ-patched-final}.  The four-dimensional equations do not see a discontinuity because the gauge kinetic matrix changes simultaneously.  Thus the EMD slice is a local truncation of a globally non-geometric vector bundle.

\section{Quantisation and Thermodynamics}
\label{sec:quant-thermo}

The construction has two independent integral structures.  First, the T-fold monodromy is integral,
\begin{equation}
  g_Q(n)\in O(3,3;\ZZ),
  \qquad n\in\ZZ.
  \label{eq:nquant-final}
\end{equation}
This is a global statement about the compactification.  Second, the black-hole charge vectors lie in the integral type-IIB charge lattice.  In the Narain description,
\begin{equation}
  Q,P\in\Gamma^{6,6},
  \qquad
  Q_Q,P_Q\in\Gamma^{6,6}.
  \label{eq:NarainQuant-final}
\end{equation}
In the EMD slice this reduces to
\begin{equation}
  P_Q=pP_0,
  \qquad
  Q_Q=eQ_0,
  \qquad
  p,e\in\ZZ,
  \label{eq:EMDquant-final}
\end{equation}
where $P_0$ and $Q_0$ are fixed charge units determined by the compactification.
The extremal condition of the EMD branch is
\begin{equation}
  r_+=r_-
  \quad\Longleftrightarrow\quad
  M^2+\Sigma^2=P_Q^2+Q_Q^2.
  \label{eq:extcond-final}
\end{equation}
For mutually BPS signs,
\begin{equation}
  M_{\rm ext}=\frac{|P_Q|+|Q_Q|}{\sqrt2},
  \qquad
  \Sigma_{\rm ext}=\frac{|P_Q|-|Q_Q|}{\sqrt2}.
  \label{eq:Mextext-final}
\end{equation}
The extremal horizon radius squared is
\begin{equation}
  R_H^2=M_{\rm ext}^2-\Sigma_{\rm ext}^2=2|P_QQ_Q|.
  \label{eq:RH2EMD-final}
\end{equation}
On the equal-charge sublattice,
\begin{equation}
  R_H=\Theta,
  \qquad
  M_{\rm ext}=\Theta.
  \label{eq:ThetaExt-final}
\end{equation}
Thus charge quantisation discretises the BPS mass at fixed asymptotic moduli.  It does not discretise the full classical non-extremal ADM parameter.  In the two-derivative effective theory the excess energy above the BPS bound is continuous.
For the general EMD black hole the outer horizon is at $r=r_+$, and the thermodynamic quantities are defined in the standard Bekenstein-Hawking framework \cite{Bekenstein:1973ur,Bardeen:1973gs,Hawking:1974sw,Wald:1993nt,Iyer:1994ys}.  The area, entropy and temperature are
\begin{equation}
  A_H=4\pi(r_+^2-\Sigma^2),
  \qquad
  S_{\rm BH}=\frac{\pi(r_+^2-\Sigma^2)}{G_4},
  \qquad
  T_H=\frac{r_+-r_-}{4\pi(r_+^2-\Sigma^2)}.
  \label{eq:EMDthermo-final}
\end{equation}
The first law is
\begin{equation}
  dM=T_HdS_{\rm BH}+\Phi_e\,dQ_Q+\Phi_m\,dP_Q-\Sigma\,d\phi_\infty.
  \label{eq:firstlaw-final}
\end{equation}
At fixed asymptotic scalar the final term is absent.  In the regular extremal branch,
\begin{equation}
  S_{\rm ext}=\frac{2\pi|P_QQ_Q|}{G_4}.
  \label{eq:extentropyEMD-final}
\end{equation}

For the equal-charge RN-like sublattice,
\begin{equation}
  r_H=M+\sqrt{M^2-\Theta^2},
  \qquad M\geq\Theta.
  \label{eq:rH-final}
\end{equation}
The temperature and entropy are
\begin{equation}
  T_H=\frac1{4\pi r_H}\left(1-\frac{\Theta^2}{r_H^2}\right),
  \qquad
  S_{\rm BH}=\frac{\pi r_H^2}{G_4}.
  \label{eq:RNthermo-final}
\end{equation}
Writing $M=(r_H^2+\Theta^2)/(2r_H)$, the fixed-$\Theta$ heat capacity is
\begin{equation}
  C_\Theta=\left(\frac{\partial M}{\partial T_H}\right)_\Theta
  =\frac{2\pi r_H^2}{G_4}\frac{r_H^2-\Theta^2}{3\Theta^2-r_H^2}.
  \label{eq:heat-final}
\end{equation}
Since $r_H\geq\Theta$, the numerator is non-negative.  Therefore
\begin{equation}
  C_\Theta>0\quad \text{for}\quad \Theta<r_H<\sqrt3\,\Theta,
  \qquad
  C_\Theta<0\quad \text{for}\quad r_H>\sqrt3\,\Theta.
  \label{eq:heat-sign-final}
\end{equation}
This is a local canonical statement.  No global thermodynamic stability is inferred without specifying the ensemble and the full Hessian in charge space. The equal-charge branch obeys the standard algebraic relation
\begin{equation}
  M=\frac{r_H^2+\Theta^2}{2r_H}.
  \label{eq:M-rH-main}
\end{equation}
The potential conjugate to the charge scale is
\begin{equation}
  \Phi_\Theta=\frac{\Theta}{r_H}.
  \label{eq:PhiTheta-main}
\end{equation}
Using \eqref{eq:RNthermo-final}, one obtains
\begin{equation}
  M=2T_HS_{\rm BH}+\Phi_\Theta\Theta,
  \label{eq:Smarr-main}
\end{equation}
which is the usual four-dimensional Smarr relation for the RN-like sublattice.  Taking the exterior derivative gives the fixed-frame first law
\begin{equation}
  dM=T_HdS_{\rm BH}+\Phi_\Theta d\Theta.
  \label{eq:firstlaw-RN-main}
\end{equation}
The generic EMD first law contains separate electric, magnetic and scalar-boundary terms as in \eqref{eq:firstlaw-final}.  The equal-charge branch is a special constant-scalar slice where those terms combine into the single charge-scale variation.
In the full type-IIB orbit the BPS bound is expressed through the central charge,
\begin{equation}
  M\geq |Z(Q_Q,P_Q;\mathcal M_\infty)|,
  \qquad
  M_{\rm BPS}=|Z(Q_Q,P_Q;\mathcal M_\infty)|.
  \label{eq:BPSbound-final}
\end{equation}
At the attractor point,
\begin{equation}
  |Z_*|^2=\sqrt{|\Delta(Q_Q,P_Q)|}.
  \label{eq:ZDelta-final}
\end{equation}
Thus the BPS mass lattice at fixed moduli is the image of the integral charge lattice under the central-charge map.

\section{Supersymmetry and Stability}
\label{sec:susy-stab}

The minimal EMD action is a bosonic slice.  Supersymmetry is most cleanly discussed in the type-IIB seed frame and then transported to the $Q$-frame by exact duality.  The F1-P-NS5-KKM system is a standard mutually BPS four-charge orbit for a compatible choice of orientations and signs.  The paper uses central-charge saturation and duality invariance of the index.  A representative projector computation in one conventional orientation is given in Appendix \ref{app:susy-projectors}.
The extremal T-fold representative is BPS because it is an exact duality representative of the type-IIB BPS seed.  If $S(g_Q)$ denotes the spinorial lift of the active duality transformation in the relevant representation, the transformed projector algebra is obtained by conjugation,
\begin{equation}
  \Pi_A^{(Q)}=S(g_Q)\Pi_AS(g_Q)^{-1}.
  \label{eq:projector-conj-final}
\end{equation}
An invertible conjugation does not change the rank of the common eigenspace.  Thus the number of preserved indexed supercharges is the same as in the seed frame.  In four-dimensional language this is the central-charge statement \eqref{eq:BPSbound-final} together with the attractor relation \eqref{eq:ZDelta-final}.

No local Killing-spinor equation containing $Q_1{}^{23}$ as if it were a radial source is used.  Such an equation would confuse a local representative of global monodromy with a local stress tensor.  The BPS statement is inherited from exact duality of the charge orbit.
We now record the conservative stability statement, following the standard Regge-Wheeler-Zerilli and charged-black-hole perturbation logic \cite{Regge:1957td,Zerilli:1970se,Moncrief:1974gw,Chandrasekhar:1975zza,Kodama:2003jz}.  A full perturbation analysis of the doubled T-fold background would require coupled fluctuations of the external metric, scalar matrix, doubled vectors, doubled dilaton and fermions.  We prove only the absence of tachyonic bound states for a minimally coupled neutral probe scalar in the exterior of the equal-charge branch.
Let
\begin{equation}
  \Phi_{\rm probe}(t,r,\theta,\varphi)=\frac{\psi_\ell(r)}{r}Y_{\ell m}(\theta,\varphi)e^{-i\omega t},
  \qquad
  \frac{dr_*}{dr}=\frac1{f(r)}.
  \label{eq:scalaransatz-final}
\end{equation}
The wave equation $\nabla^2\Phi_{\rm probe}=0$ becomes
\begin{equation}
  \frac{d^2\psi_\ell}{dr_*^2}+\bigl[\omega^2-V_\ell(r)\bigr]\psi_\ell=0,
  \label{eq:schrod-final}
\end{equation}
with
\begin{equation}
  V_\ell(r)=f(r)\left[\frac{\ell(\ell+1)}{r^2}+\frac{f'(r)}{r}\right].
  \label{eq:potential-general-final}
\end{equation}
For $f(r)=1-2M/r+\Theta^2/r^2$,
\begin{equation}
  V_\ell(r)=f(r)\left[\frac{\ell(\ell+1)}{r^2}+\frac{2M}{r^3}-\frac{2\Theta^2}{r^4}\right].
  \label{eq:potential-RN-final}
\end{equation}
For $r\geq r_H$, $f(r)\geq0$, and
\begin{align}
  Mr-\Theta^2&\geq Mr_H-\Theta^2
  =M\left(M+\sqrt{M^2-\Theta^2}\right)-\Theta^2 \notag\\
  &=(M^2-\Theta^2)+M\sqrt{M^2-\Theta^2}\geq0.
  \label{eq:positivity-final}
\end{align}
Thus the $s$-wave bracket in \eqref{eq:potential-RN-final} is non-negative, and higher partial waves add a positive centrifugal term.  The exterior operator is non-negative with the usual black-hole boundary conditions.  There are no neutral probe-scalar modes with $\omega^2<0$ in this sector.  This statement is not a claim of full coupled DFT stability.
The asymptotic potential behaves as
\begin{equation}
  V_\ell(r)=\frac{\ell(\ell+1)}{r^2}+\frac{2M}{r^3}+O(r^{-4}),
  \label{eq:tail-final}
\end{equation}
so a neutral probe has the same leading long-range tail as in the ordinary Reissner-Nordstrom problem.
The positivity argument is insensitive to the non-geometric patching because the probe is neutral under the doubled vector bundle.  A probe carrying momentum or winding charge would couple to the local system and would see patch-dependent gauge potentials.  Such a problem is more subtle: the mode functions must be sections of the duality bundle rather than ordinary scalar functions.  This is why the statement made here is deliberately restricted to neutral probes.  Charged probes or doubled-vector fluctuations require a separate coupled analysis.

\section{Near-Horizon Duality Map}
\label{sec:nh-micro}

On the extremal EMD branch,
\begin{equation}
  r_+=r_-=M,
  \qquad
  R_H^2=M^2-\Sigma^2=2|P_QQ_Q|.
  \label{eq:RH-EMD-final}
\end{equation}
For positive charges, set
\begin{equation}
  r=M+\lambda\rho,
  \qquad
  t=\frac{R_H^2}{\lambda}\tau,
  \qquad
  \lambda\rightarrow0.
  \label{eq:NHscaling-final}
\end{equation}
The four-dimensional metric tends to
\begin{equation}
  ds_4^2\rightarrow R_H^2\left[-\rho^2d\tau^2+\frac{d\rho^2}{\rho^2}+d\Omega_2^2\right].
  \label{eq:NHmetric-final}
\end{equation}
The horizon scalar is
\begin{equation}
  e^{2\phi_H}=\left|\frac{P_Q}{Q_Q}\right|,
  \qquad
  e^{-2\phi_H}=\left|\frac{Q_Q}{P_Q}\right|.
  \label{eq:horizon-scalar-final}
\end{equation}
Including compact directions, the local near-horizon geometry is
\begin{equation}
  \AdS_2(R_H)\times S^2(R_H)\times T_Q^3\times T_s^3.
  \label{eq:NHlocal-final}
\end{equation}
The active $T_Q^3$ is globally patched by the monodromy \eqref{eq:local-monodromy-final}, so the near-horizon compactification is a T-fold.  This local $\AdS_2$ factor is the same geometric ingredient that appears in the entropy-function and near-horizon holography literature \cite{Strominger:1998yg,Maldacena:1997re,Sen:2008yk}.
The entropy-function calculation gives an independent derivation of the EMD entropy \cite{Sen:2005wa,Sen:2008yk}.  Use the near-horizon ansatz
\begin{equation}
  ds^2=v_1\left(-\rho^2dt^2+\frac{d\rho^2}{\rho^2}\right)+v_2d\Omega_2^2,
  \label{eq:EFmetric-final}
\end{equation}
\begin{equation}
  F_Q=\mathfrak e\,dt\wedge d\rho+P_Q\sin\theta\,d\theta\wedge d\varphi,
  \qquad
  \phi=u.
  \label{eq:EFfields-final}
\end{equation}
The Ricci scalar and field strength are
\begin{equation}
  R=-\frac2{v_1}+\frac2{v_2},
  \qquad
  F_Q^2=2\left(-\frac{\mathfrak e^2}{v_1^2}+\frac{P_Q^2}{v_2^2}\right).
  \label{eq:EFRF-final}
\end{equation}
The angular integral of the Lagrangian is
\begin{equation}
  f(v_1,v_2,u,\mathfrak e,P_Q)=\frac1{4G_4}v_1v_2\left[-\frac2{v_1}+\frac2{v_2}
  -2e^{-2u}\left(-\frac{\mathfrak e^2}{v_1^2}+\frac{P_Q^2}{v_2^2}\right)\right].
  \label{eq:entropy-f-final}
\end{equation}
The conjugate electric charge is
\begin{equation}
  q_{\rm Sen}=\frac{\partial f}{\partial\mathfrak e}=\frac1{G_4}\frac{v_2}{v_1}e^{-2u}\mathfrak e.
  \label{eq:qSen-final}
\end{equation}
With the convention \eqref{eq:EMDcharges-final}, $Q_Q=G_4q_{\rm Sen}$.  Sen's entropy function is
\begin{equation}
  \cE=2\pi(q_{\rm Sen}\mathfrak e-f).
  \label{eq:entropy-function-final}
\end{equation}
Extremisation gives
\begin{equation}
  \frac{\mathfrak e^2}{v_1^2}=\frac{P_Q^2}{v_2^2},
  \qquad
  e^{-2u}=\left|\frac{Q_Q}{P_Q}\right|,
  \qquad
  v_1=v_2=2|P_QQ_Q|.
  \label{eq:EFextremum-final}
\end{equation}
The extremised value is
\begin{equation}
  S_{\rm Wald}=\cE_{\rm ext}=\frac{\pi v_2}{G_4}=\frac{2\pi|P_QQ_Q|}{G_4}.
  \label{eq:SWald-final}
\end{equation}
On the equal-charge sublattice this reduces to $\pi\Theta^2/G_4$.  In the full type-IIB orbit the entropy-function statement is
\begin{equation}
  S_{\rm Wald}=\frac{\pi}{G_4}\sqrt{|\Delta(Q_Q,P_Q)|}.
  \label{eq:SWaldDelta-final}
\end{equation}

The microscopic counting is performed in the geometric type-IIB F1-P-NS5-KKM frame, where the weak-coupling index and related brane-counting arguments are under control \cite{Strominger:1996sh,Callan:1996dv,Maldacena:1996ix,Horowitz:1996fn,Johnson:1996ga,Maldacena:1996gb,Maldacena:1997de,Mathur:2005zp,Bena:2007kg}.  The large-charge index gives \eqref{eq:microIIB-main}.  Applying the integral duality transformation gives a charge vector in the $Q$-frame, but exact duality preserves the index:
The equality of the index can be expressed in orbit language.  Let $\mathcal O(Q,P)$ denote the orbit of the charge pair under the exact integral duality group.  The indexed degeneracy is constant on $\mathcal O(Q,P)$ for the BPS sector under consideration.  Since $(Q_Q,P_Q)$ lies in the same orbit as $(Q,P)$, the logarithm of the degeneracy is the same.  The saddle-point approximation at large charge then gives the same square-root growth as in the geometric frame.

\begin{equation}
  d(Q_Q,P_Q)=d(Q,P),
  \qquad
  \Delta(Q_Q,P_Q)=\Delta(Q,P).
  \label{eq:index-pres-final}
\end{equation}
Therefore the microscopic entropy of the T-fold representative equals its macroscopic entropy.
It is important to separate three frames.  The counting frame is the F1-P-NS5-KKM frame, where the BPS index is computed.  The T-fold $Q$-frame is obtained by applying $h_Q$ to the charge local system.  An exotic display frame is a further duality-related representative in which some five-brane or monopole charges appear as monodromy charges.  A beta-transform may be represented as a T-duality conjugate of a large $B$-field gauge transformation,
\begin{equation}
  \beta^{23}\text{-shift}=T_2T_3\circ(B_{23}\text{-shift})\circ T_3T_2.
  \label{eq:beta-conjugate-final}
\end{equation}
The counting-frame brane orientation need not itself be the orientation in which an NS5 is immediately dualised to a $5^2_2$ brane.  The exotic display frame is first reached by ordinary T-dualities to a representative with two transverse isometries along the beta directions; only then does the conjugation by T-dualities produce the $5^2_2$-type monodromy representative.  This exotic-frame description is an interpretation of the same charge orbit and uses the standard T-duality relation between geometric five-branes, monopoles and exotic monodromy branes \cite{Obers:1998fb,Tong:2002rq,Harvey:2005ab,deBoer:2012ma,Kimura:2014wsa}.  No independent exotic-brane Cardy calculation is claimed.

\section{Conclusion}
\label{sec:conclusion}

A T-fold representative of a four-dimensional dyonic black-hole charge orbit has been constructed in doubled type-IIB theory.  The central point is the separation between the global non-geometric monodromy of the active compactification and the conserved four-dimensional charges that source the external black-hole geometry.  The monodromy determines the duality patching of the active doubled torus and of the charge local system.  The external Reissner-Nordstrom-like behaviour is produced by ordinary four-dimensional electric and magnetic field strengths in that local system, not by treating a purely internal algebraic flux component as an external stress tensor.

The geometric starting point is the type-IIB F1-P-NS5-KKM toroidal black-hole orbit.  The active parabolic T-duality transformation gives a non-geometric representative of the same orbit.  The local field equations, attractor mechanism, entropy and indexed microscopic degeneracy are unchanged because the construction uses exact integral duality.  The full entropy is controlled by the $O(6,6)$ invariant of the NS-NS/$N=4$ type-IIB charge sublattice, while the STU expression is a convenient representative on a four-charge sublattice.

A minimal Einstein-Maxwell-dilaton slice was used to display the local black-hole fields explicitly.  The generic two-charge branch carries a running scalar.  The equal-charge branch has constant scalar and gives the Reissner-Nordstrom-like metric that preserves the original black-hole interpretation.  The extremal near-horizon geometry is locally the standard anti-de Sitter two-space and two-sphere times compact torus factors, while the active compact factor is globally patched by T-duality.

The microscopic count is inherited from the type-IIB geometric brane frame.  The exotic-brane language is a duality-frame description of the same charge orbit rather than an independent counting assumption.  Natural extensions include the full coupled perturbation problem in the doubled theory, higher-derivative Wald entropy in the T-fold frame, the boundary dynamics of the doubled near-horizon compactification, and charge-lattice constraints related to swampland bounds \cite{Wald:1993nt,Iyer:1994ys,Strominger:1998yg,ArkaniHamed:2006dz,Ooguri:2006in}.
\section*{Conflict of Interest}
The authors declare that there are no conflicts of interest related to the content, authorship, or publication of this manuscript.
 \section*{Funding Statement }
 No external funding or institutional support was involved in the preparation of this work.
\section*{Data Availability statements }
No data is associated with this manuscript.
 \bibliographystyle{ytphys}
\bibliography{references}

@article{Hull:2004in,
  author        = {Hull, C. M.},
  title         = {A geometry for non-geometric string backgrounds},
  journal       = {JHEP},
  volume        = {10},
  pages         = {065},
  year          = {2005},
  doi           = {10.1088/1126-6708/2005/10/065}
}

@article{Dabholkar:2005ve,
  author        = {Dabholkar, A. and Hull, C.},
  title         = {Generalised {T}-duality and non-geometric backgrounds},
  journal       = {JHEP},
  volume        = {05},
  pages         = {009},
  year          = {2006},
  doi           = {10.1088/1126-6708/2006/05/009}
}

@article{Hohm:2010pp,
  author        = {Hohm, O. and Hull, C. and Zwiebach, B.},
  title         = {Background independent action for double field theory},
  journal       = {JHEP},
  volume        = {07},
  pages         = {016},
  year          = {2010},
  doi           = {10.1007/JHEP07(2010)016} 
}

@article{Hohm:2010jy,
  author        = {Hohm, O. and Hull, C. and Zwiebach, B.},
  title         = {Generalized metric formulation of double field theory},
  journal       = {JHEP},
  volume        = {08},
  pages         = {008},
  year          = {2010},
  doi           = {10.1007/JHEP08(2010)008}
}

@article{Aldazabal:2013sca,
  author        = {Aldazabal, G. and Marques, D. and Nunez, C.},
  title         = {Double field theory: a pedagogical review},
  journal       = {Class. Quant. Grav.},
  volume        = {30},
  pages         = {163001},
  year          = {2013},
  doi           = {10.1088/0264-9381/30/16/163001}
}

@article{Shelton:2005cf,
  author        = {Shelton, J. and Taylor, W. and Wecht, B.},
  title         = {Nongeometric flux compactifications},
  journal       = {JHEP},
  volume        = {10},
  pages         = {085},
  year          = {2005},
  doi           = {10.1088/1126-6708/2005/10/085}
}

@article{Andriot:2011uh,
  author        = {Andriot, D. and Larfors, M. and Lust, D. and Patalong, P.},
  title         = {A ten-dimensional action for non-geometric fluxes},
  journal       = {JHEP},
  volume        = {09},
  pages         = {134},
  year          = {2011},
  doi           = {10.1007/JHEP09(2011)134} 
}

@article{Plauschinn:2018wbo,
  author        = {Plauschinn, E.},
  title         = {Non-geometric backgrounds in string theory},
  journal       = {Phys. Rept.},
  volume        = {798},
  pages         = {1-122},
  year          = {2019},
  doi           = {10.1016/j.physrep.2018.12.002}
}

@article{deBoer:2012ma,
  author        = {de Boer, J. and Shigemori, M.},
  title         = {Exotic branes in string theory},
  journal       = {Phys. Rept.},
  volume        = {532},
  pages         = {65-118},
  year          = {2013},
  doi           = {10.1016/j.physrep.2013.07.003} 
}

@article{Kimura:2014wsa,
  author        = {Kimura, T. and Sasaki, S. and Yata, M.},
  title         = {World-volume effective actions of exotic five-branes},
  journal       = {JHEP},
  volume        = {07},
  pages         = {127},
  year          = {2014},
  doi           = {10.1007/JHEP07(2014)127} 
}

@article{Gibbons:1987ps,
  author  = {Gibbons, G. W. and Maeda, K.},
  title   = {Black holes and membranes in higher dimensional theories with dilaton fields},
  journal = {Nucl. Phys. B},
  volume  = {298},
  pages   = {741-775},
  year    = {1988},
  doi     = {10.1016/0550-3213(88)90006-5}
}

@article{Kallosh:1992ii,
  author        = {Kallosh, R. and Peet, A. W.},
  title         = {Dilaton black holes near the horizon},
  journal       = {Phys. Rev. D},
  volume        = {46},
  pages         = {R5223-R5227},
  year          = {1992},
  doi           = {10.1103/PhysRevD.46.R5223}
}

@article{Kallosh:1992eh,
  author        = {Kallosh, R. and Ortin, T. and Peet, A. W.},
  title         = {Entropy and action of dilaton black holes},
  journal       = {Phys. Rev. D},
  volume        = {47},
  pages         = {5400-5407},
  year          = {1993},
  doi           = {10.1103/PhysRevD.47.5400} 
}

@article{Ferrara:1995ih,
  author        = {Ferrara, S. and Kallosh, R. and Strominger, A.},
  title         = {{N=2} extremal black holes},
  journal       = {Phys. Rev. D},
  volume        = {52},
  pages         = {R5412-R5416},
  year          = {1995},
  doi           = {10.1103/PhysRevD.52.R5412} 
}

@article{Cvetic:1995uj, 
  title = {Dyonic {BPS} saturated black holes of heterotic string on a six-torus},
  author = {Cveti\ifmmode \check{c}\else \v{c}\fi{}, Mirjam and Youm, Donam},
  journal = {Phys. Rev. D},
  volume = {53},
  issue = {2},
  pages = {R584(R)-R588(R)},
  numpages = {0},
  year = {1996},
  doi = {10.1103/PhysRevD.53.R584}
}

@article{Strominger:1996sh,
  author        = {Strominger, A. and Vafa, C.},
  title         = {Microscopic origin of the Bekenstein-Hawking entropy},
  journal       = {Phys. Lett. B},
  volume        = {379},
  pages         = {99-104},
  year          = {1996},
  doi           = {10.1016/0370-2693(96)00345-0} 
}

@article{Johnson:1996ga,
  author        = {Johnson, C. V. and Khuri, R. R. and Myers, R. C.},
  title         = {Entropy of {4D} extremal black holes},
  journal       = {Phys. Lett. B},
  volume        = {378},
  pages         = {78-86},
  year          = {1996},
  doi           = {10.1016/0370-2693(96)00383-8} 
}

@article{Maldacena:1996gb,
  author        = {Maldacena, J. M. and Strominger, A.},
  title         = {Statistical entropy of four-dimensional extremal black holes},
  journal       = {Phys. Rev. Lett.},
  volume        = {77},
  pages         = {428-429},
  year          = {1996},
  doi           = {10.1103/PhysRevLett.77.428}
}

@article{Sen:2005wa,
  author        = {Sen, A.},
  title         = {Black hole entropy function and the attractor mechanism in higher derivative gravity},
  journal       = {JHEP},
  volume        = {09},
  pages         = {038},
  year          = {2005},
  doi           = {10.1088/1126-6708/2005/09/038} 
}

@article{Sen:2008yk,
  author        = {Sen, A.},
  title         = {Entropy function and {AdS2/CFT1} correspondence},
  journal       = {JHEP},
  volume        = {11},
  pages         = {075},
  year          = {2008},
  doi           = {10.1088/1126-6708/2008/11/075}
}

@article{Arvanitakis:2016xjz,
  author        = {Arvanitakis, A. S. and Blair, C. D. A.},
  title         = {Black hole thermodynamics, stringy dualities and double field theory},
  journal       = {Class. Quant. Grav.},
  volume        = {34},
  pages         = {055001},
  year          = {2017},
  doi           = {10.1088/1361-6382/aa5a59} 
}

@incollection{Bena:2007kg,
  author        = {Bena, I. and Warner, N. P.},
  title         = {Black holes, black rings and their microstates},
  booktitle     = {Lecture Notes in Physics},
  volume        = {755},
  pages         = {1-92},
  year          = {2008},
  publisher     = {Springer},
  doi           = {10.1007/978-3-540-79523-0_1} 
}

@article{Maharana:1992my,
  author        = {Maharana, J. and Schwarz, J. H.},
  title         = {Noncompact symmetries in string theory},
  journal       = {Nucl. Phys. B},
  volume        = {390},
  pages         = {3-32},
  year          = {1993},
  doi           = {10.1016/0550-3213(93)90387-5}
}

@article{Giveon:1994fu,
  author        = {Giveon, A. and Porrati, M. and Rabinovici, E.},
  title         = {Target space duality in string theory},
  journal       = {Phys. Rept.},
  volume        = {244},
  pages         = {77-202},
  year          = {1994},
  doi           = {10.1016/0370-1573(94)90070-1}
}

@article{Kikkawa:1984cp,
  author = "Kikkawa, Keiji and Yamasaki, Masami",
    title = "{Casimir Effects in Superstring Theories}",
    reportNumber = "OU-HET 61",
    doi = "10.1016/0370-2693(84)90423-4",
    journal = "Phys. Lett. B",
    volume = "149",
    pages = "357-360",
    year = "1984"
}

@article{Sakai:1986vg,
  author        = {Sakai, N. and Senda, I.},
  title         = {Vacuum Energies of String Compactified on Torus},
  journal       = {Prog. Theor. Phys.},
  volume        = {75},
  pages         = {692},
  year          = {1986},
  doi           = {10.1143/PTP.75.692}
}

@article{Narain:1985jj,
  author        = {Narain, K. S.},
  title         = {New Heterotic String Theories in Uncompactified Dimensions Less Than Ten},
  journal       = {Phys. Lett. B},
  volume        = {169},
  pages         = {41-46},
  year          = {1986},
  doi           = {10.1016/0370-2693(86)90682-9}
}

@article{Narain:1986am,
  author        = {Narain, K. S. and Sarmadi, M. H. and Witten, E.},
  title         = {A Note on Toroidal Compactification of Heterotic String Theory},
  journal       = {Nucl. Phys. B},
  volume        = {279},
  pages         = {369-379},
  year          = {1987},
  doi           = {10.1016/0550-3213(87)90001-0}
}

@article{Buscher:1987sk,
  author        = {Buscher, T. H.},
  title         = {A Symmetry of the String Background Field Equations},
  journal       = {Phys. Lett. B},
  volume        = {194},
  pages         = {59-62},
  year          = {1987},
  doi           = {10.1016/0370-2693(87)90769-6}
}

@article{Buscher:1987qj,
  author        = {Buscher, T. H.},
  title         = {Path Integral Derivation of Quantum Duality in Nonlinear Sigma Models},
  journal       = {Phys. Lett. B},
  volume        = {201},
  pages         = {466-472},
  year          = {1988},
  doi           = {10.1016/0370-2693(88)90602-8}
}

@article{Meissner:1991zj,
  author        = {Meissner, K. A. and Veneziano, G.},
  title         = {Symmetries of Cosmological Superstring Vacua},
  journal       = {Phys. Lett. B},
  volume        = {267},
  pages         = {33-36},
  year          = {1991},
  doi           = {10.1016/0370-2693(91)90520-Z}
}

@article{Tseytlin:1990nb,
  author        = {Tseytlin, A. A.},
  title         = {Duality Symmetric Closed String Theory and Interacting Chiral Scalars},
  journal       = {Nucl. Phys. B},
  volume        = {350},
  pages         = {395-440},
  year          = {1991},
  doi           = {10.1016/0550-3213(91)90266-Z}
}

@article{Tseytlin:1990va,
  author        = {Tseytlin, A. A.},
  title         = {Duality Symmetric String Theory and the Cosmological Constant Problem},
  journal       = {Phys. Rev. Lett.},
  volume        = {66},
  pages         = {545-548},
  year          = {1991},
  doi           = {10.1103/PhysRevLett.66.545}
}

@article{Siegel:1993xq,
  author        = {Siegel, W.},
  title         = {Two Vierbein Formalism for String Inspired Axionic Gravity},
  journal       = {Phys. Rev. D},
  volume        = {47},
  pages         = {5453-5459},
  year          = {1993},
  doi           = {10.1103/PhysRevD.47.5453} 
}

@article{Siegel:1993th,
  author        = {Siegel, W.},
  title         = {Superspace Duality in Low-Energy Superstrings},
  journal       = {Phys. Rev. D},
  volume        = {48},
  pages         = {2826-2837},
  year          = {1993},
  doi           = {10.1103/PhysRevD.48.2826} 
}

@article{Courant:1990,
  author        = {Courant, T. J.},
  title         = {Dirac Manifolds},
  journal       = {Trans. Amer. Math. Soc.},
  volume        = {319},
  pages         = {631-661},
  year          = {1990},
  doi           = {10.1090/S0002-9947-1990-0998124-1}
}

@article{Hitchin:2002,
  author        = {Hitchin, N.},
  title         = {Generalized Calabi-Yau Manifolds},
  journal       = {Quart. J. Math.},
  volume        = {54},
  number        = {3},
  pages         = {281-308},
  year          = {2003},
  doi           = {10.1093/qjmath/54.3.281} 
}

@article{Gualtieri:2011dx,
  author        = {Gualtieri, M.},
  title         = {Generalized Complex Geometry},
  journal       = {Ann. Math.},
  volume        = {174},
  number        = {1},
  pages         = {75-123},
  year          = {2011},
  doi           = {10.4007/annals.2011.174.1.3} 
}

@article{Hull:2009mi,
  author        = {Hull, C. and Zwiebach, B.},
  title         = {Double Field Theory},
  journal       = {JHEP},
  volume        = {09},
  pages         = {099},
  year          = {2009},
  doi           = {10.1088/1126-6708/2009/09/099} 
}

@article{Hohm:2010xq,
  author        = {Hohm, O. and Kwak, S. K.},
  title         = {Frame-like Geometry of Double Field Theory},
  journal       = {J. Phys. A},
  volume        = {44},
  pages         = {085404},
  year          = {2011},
  doi           = {10.1088/1751-8113/44/8/085404}
}

@article{Hohm:2011si,
  author        = {Hohm, O. and Zwiebach, B.},
  title         = {On the Riemann Tensor in Double Field Theory},
  journal       = {JHEP},
  volume        = {05},
  pages         = {126},
  year          = {2012},
  doi           = {10.1007/JHEP05(2012)126}
}

@article{Geissbuhler:2011mx,
  author        = {Geissbuhler, D.},
  title         = {Double Field Theory and $\mathcal{N}=4$ Gauged Supergravity},
  journal       = {JHEP},
  volume        = {11},
  pages         = {116},
  year          = {2011},
  doi           = {10.1007/JHEP11(2011)116}
}

@article{Grana:2012rr,
  author        = {Grana, M. and Marques, D.},
  title         = {Gauged Double Field Theory},
  journal       = {JHEP},
  volume        = {04},
  pages         = {020},
  year          = {2012},
  doi           = {10.1007/JHEP04(2012)020} 
}

@article{Berman:2011pe,
  author        = {Berman, D. S. and Perry, M. J.},
  title         = {Generalized Geometry and ${M}$ Theory},
  journal       = {JHEP},
  volume        = {06},
  pages         = {074},
  year          = {2011},
  doi           = {10.1007/JHEP06(2011)074} 
}

@article{Hohm:2013pua,
  author        = {Hohm, O. and Samtleben, H.},
  title         = {Exceptional Field Theory {I}: E6(6) Covariant Form of {M}-Theory and Type {IIB}},
  journal       = {Phys. Rev. D},
  volume        = {89},
  pages         = {066016},
  year          = {2014},
  doi           = {10.1103/PhysRevD.89.066016}
}

@article{Grana:2005jc,
  author        = {Grana, M.},
  title         = {Flux Compactifications in String Theory: A Comprehensive Review},
  journal       = {Phys. Rept.},
  volume        = {423},
  pages         = {91-158},
  year          = {2006},
  doi           = {10.1016/j.physrep.2005.10.008}
}

@article{Douglas:2006es,
  author        = {Douglas, M. R. and Kachru, S.},
  title         = {Flux Compactification},
  journal       = {Rev. Mod. Phys.},
  volume        = {79},
  pages         = {733-796},
  year          = {2007},
  doi           = {10.1103/RevModPhys.79.733}
}

@article{Kachru:2002he,
  author        = {Kachru, S. and Schulz, M. B. and Trivedi, S.},
  title         = {Moduli Stabilization from Fluxes in a Simple {IIB} Orientifold},
  journal       = {JHEP},
  volume        = {10},
  pages         = {007},
  year          = {2003},
  doi           = {10.1088/1126-6708/2003/10/007} 
}

@article{Hellerman:2002ax,
  author        = {Hellerman, S. and McGreevy, J. and Williams, B.},
  title         = {Geometric Constructions of Nongeometric String Theories},
  journal       = {JHEP},
  volume        = {01},
  pages         = {024},
  year          = {2004},
  doi           = {10.1088/1126-6708/2004/01/024} 
}

@article{Flournoy:2004vn,
  title = {Constructing nongeometric vacua in string theory},
journal = {Nuclear Physics B},
volume = {706},
number = {1},
pages = {127-149},
year = {2005},
issn = {0550-3213},
doi = {https://doi.org/10.1016/j.nuclphysb.2004.11.005},
author = {Alex Flournoy and Brian Wecht and Brook Williams}
}

@article{Hull:2009sg,
 author = "Hull, C. M. and Reid-Edwards, R. A.",
    title = "{Flux compactifications of string theory on twisted tori}", 
    doi = "10.1002/prop.200900076",
    journal = "Fortsch. Phys.",
    volume = "57",
    pages = "862-894",
    year = "2009"
}

@article{DallAgata:2005zlf,
 author = "D'Auria, Riccardo and Ferrara, Sergio and Vaula, Silvia",
    title = "{$\mathcal{N}=4$ gauged supergravity and a {IIB} orientifold with fluxes}", 
    doi = "10.1088/1367-2630/4/1/371",
    journal = "New J. Phys.",
    volume = "4",
    pages = "71",
    year = "2002"
}

@article{Aldazabal:2006up,
  author        = {Aldazabal, G. and Font, A. and Marques, D. and Nunez, C.},
  title         = {More Dual Fluxes and Moduli Fixing},
  journal       = {JHEP},
  volume        = {05},
  pages         = {070},
  year          = {2006},
  doi           = {10.1088/1126-6708/2006/05/070} 
}

@article{Blumenhagen:2011ph,
  author        = {Blumenhagen, R. and Deser, A. and Lust, D. and Plauschinn, E. and Rennecke, F.},
  title         = {Non-Geometric Fluxes, Asymmetric Strings and Nonassociative Geometry},
  journal       = {J. Phys. A},
  volume        = {44},
  pages         = {385401},
  year          = {2011},
  doi           = {10.1088/1751-8113/44/38/385401} 
}

@article{Hassler:2014sba,
  author        = {Hassler, F. and Lust, D.},
  title         = {Consistent Compactification of Double Field Theory on Non-Geometric Flux Backgrounds},
  journal       = {JHEP},
  volume        = {05},
  pages         = {085},
  year          = {2014},
  doi           = {10.1007/JHEP05(2014)085} 
}

@article{Obers:1998fb,
  author        = {Obers, N. A. and Pioline, B.},
  title         = {{U}-Duality and {M}-Theory},
  journal       = {Phys. Rept.},
  volume        = {318},
  pages         = {113-225},
  year          = {1999},
  doi           = {10.1016/S0370-1573(99)00004-6} 
}

@article{Tong:2002rq,
 author = "Tong, David",
    title = "{NS5-branes, T duality and world sheet instantons}", 
    doi = "10.1088/1126-6708/2002/07/013",
    journal = "JHEP",
    volume = "07",
    pages = "013",
    year = "2002"
}

@article{Harvey:2005ab,
  author        = {Harvey, J. A. and Jensen, S.},
  title         = {Worldsheet Instanton Corrections to the Kaluza-Klein Monopole},
  journal       = {JHEP},
  volume        = {10},
  pages         = {028},
  year          = {2005},
  doi           = {10.1088/1126-6708/2005/10/028} 
}

@article{Garfinkle:1990qj,
  author        = {Garfinkle, D. and Horowitz, G. T. and Strominger, A.},
  title         = {Charged Black Holes in String Theory},
  journal       = {Phys. Rev. D},
  volume        = {43},
  pages         = {3140},
  year          = {1991},
  doi           = {10.1103/PhysRevD.43.3140}
}

@article{Horowitz:1991cd,
  author        = {Horowitz, G. T. and Strominger, A.},
  title         = {Black Strings and P-Branes},
  journal       = {Nucl. Phys. B},
  volume        = {360},
  pages         = {197-209},
  year          = {1991},
  doi           = {10.1016/0550-3213(91)90440-9}
}

@article{Sen:1992ua,
  author        = {Sen, A.},
  title         = {Rotating Charged Black Hole Solution in Heterotic String Theory},
  journal       = {Phys. Rev. Lett.},
  volume        = {69},
  pages         = {1006-1009},
  year          = {1992},
  doi           = {10.1103/PhysRevLett.69.1006} 
}

@article{Duff:1995sm,
 title = {Four-dimensional string/ string/ string triality},
journal = {Nuclear Physics B},
volume = {459},
number = {1},
pages = {125-159},
year = {1996},
issn = {0550-3213},
doi = {https://doi.org/10.1016/0550-3213(95)00555-2}, 
author = {M.J. Duff and James T. Liu and J. Rahmfeld},
}

@article{Kallosh:1992ap,
  author        = {Kallosh, R. and Linde, A. D. and Ortin, T. and Peet, A. W. and Van Proeyen, A.},
  title         = {Supersymmetry as a Cosmic Censor},
  journal       = {Phys. Rev. D},
  volume        = {46},
  pages         = {5278-5302},
  year          = {1992},
  doi           = {10.1103/PhysRevD.46.5278} 
}

@article{Ferrara:1996dd,
  author        = {Ferrara, S. and Kallosh, R.},
  title         = {Supersymmetry and Attractors},
  journal       = {Phys. Rev. D},
  volume        = {54},
  pages         = {1514-1524},
  year          = {1996},
  doi           = {10.1103/PhysRevD.54.1514} 
}

@article{Ferrara:1996um,
  author        = {Ferrara, S. and Kallosh, R.},
  title         = {Universality of Supersymmetric Attractors},
  journal       = {Phys. Rev. D},
  volume        = {54},
  pages         = {1525-1534},
  year          = {1996},
  doi           = {10.1103/PhysRevD.54.1525} 
}

@article{Strominger:1996kf,
  author        = {Strominger, A.},
  title         = {Macroscopic Entropy of {N=2} Extremal Black Holes},
  journal       = {Phys. Lett. B},
  volume        = {383},
  pages         = {39-43},
  year          = {1996},
  doi           = {10.1016/0370-2693(96)00711-3}
}

@article{Denef:2000nb,
  author        = {Denef, F.},
  title         = {Supergravity Flows and D-Brane Stability},
  journal       = {JHEP},
  volume        = {08},
  pages         = {050},
  year          = {2000},
  doi           = {10.1088/1126-6708/2000/08/050} 
}

@article{Maldacena:1997de,
  author        = {Maldacena, J. M. and Strominger, A. and Witten, E.},
  title         = {Black Hole Entropy in {M}-Theory},
  journal       = {JHEP},
  volume        = {12},
  pages         = {002},
  year          = {1997},
  doi           = {10.1088/1126-6708/1997/12/002},
  eprint        = {hep-th/9711053},
  archivePrefix = {arXiv},
  primaryClass  = {hep-th}
}

@article{Callan:1996dv,
  author        = {Callan, C. G. and Maldacena, J. M.},
  title         = {D-Brane Approach to Black Hole Quantum Mechanics},
  journal       = {Nucl. Phys. B},
  volume        = {472},
  pages         = {591-610},
  year          = {1996},
  doi           = {10.1016/0550-3213(96)00225-8} 
}

@article{Maldacena:1996ix,
  author        = {Maldacena, J. M. and Susskind, L.},
  title         = {D-Branes and Fat Black Holes},
  journal       = {Nucl. Phys. B},
  volume        = {475},
  pages         = {679-690},
  year          = {1996},
  doi           = {10.1016/0550-3213(96)00323-9} 
}

@article{Horowitz:1996fn,
  author        = {Horowitz, G. T. and Strominger, A.},
  title         = {Counting States of Near-Extremal Black Holes},
  journal       = {Phys. Rev. Lett.},
  volume        = {77},
  pages         = {2368-2371},
  year          = {1996},
  doi           = {10.1103/PhysRevLett.77.2368} 
}

@article{Breckenridge:1996is,
  author        = {Breckenridge, J. C. and Myers, R. C. and Peet, A. W. and Vafa, C.},
  title         = {D-Branes and Spinning Black Holes},
  journal       = {Phys. Lett. B},
  volume        = {391},
  pages         = {93-98},
  year          = {1997},
  doi           = {10.1016/S0370-2693(96)01460-8} 
}

@article{Mathur:2005zp,
  author        = {Mathur, S. D.},
  title         = {The Fuzzball Proposal for Black Holes: An Elementary Review},
  journal       = {Fortsch. Phys.},
  volume        = {53},
  pages         = {793-827},
  year          = {2005},
  doi           = {10.1002/prop.200410203}
}

@article{Bekenstein:1973ur,
  author        = {Bekenstein, J. D.},
  title         = {Black Holes and Entropy},
  journal       = {Phys. Rev. D},
  volume        = {7},
  pages         = {2333-2346},
  year          = {1973},
  doi           = {10.1103/PhysRevD.7.2333}
}

@article{Bardeen:1973gs,
  author        = {Bardeen, J. M. and Carter, B. and Hawking, S. W.},
  title         = {The Four Laws of Black Hole Mechanics},
  journal       = {Commun. Math. Phys.},
  volume        = {31},
  pages         = {161-170},
  year          = {1973},
  doi           = {10.1007/BF01645742}
}

@article{Hawking:1974sw,
  author        = {Hawking, S. W.},
  title         = {Particle Creation by Black Holes},
  journal       = {Commun. Math. Phys.},
  volume        = {43},
  pages         = {199-220},
  year          = {1975},
  doi           = {10.1007/BF02345020}
}

@article{Wald:1993nt,
  author        = {Wald, R. M.},
  title         = {Black Hole Entropy is the Noether Charge},
  journal       = {Phys. Rev. D},
  volume        = {48},
  pages         = {R3427-R3431},
  year          = {1993},
  doi           = {10.1103/PhysRevD.48.R3427} 
}

@article{Iyer:1994ys,
  author        = {Iyer, V. and Wald, R. M.},
  title         = {Some Properties of Noether Charge and a Proposal for Dynamical Black Hole Entropy},
  journal       = {Phys. Rev. D},
  volume        = {50},
  pages         = {846-864},
  year          = {1994},
  doi           = {10.1103/PhysRevD.50.846} 
}

@article{Regge:1957td,
  author        = {Regge, T. and Wheeler, J. A.},
  title         = {Stability of a Schwarzschild Singularity},
  journal       = {Phys. Rev.},
  volume        = {108},
  pages         = {1063-1069},
  year          = {1957},
  doi           = {10.1103/PhysRev.108.1063}
}

@article{Zerilli:1970se,
  author        = {Zerilli, F. J.},
  title         = {Effective Potential for Even-Parity Regge-Wheeler Gravitational Perturbation Equations},
  journal       = {Phys. Rev. Lett.},
  volume        = {24},
  pages         = {737-738},
  year          = {1970},
  doi           = {10.1103/PhysRevLett.24.737}
}

@article{Moncrief:1974gw, 
  title = {Odd-parity stability of a {Reissner-Nordstr\"om} black hole},
  author = {Moncrief, Vincent},
  journal = {Phys. Rev. D},
  volume = {9},
  issue = {10},
  pages = {2707-2709},
  numpages = {0},
  year = {1974},
  doi = {10.1103/PhysRevD.9.2707}
}

@article{Chandrasekhar:1975zza,
  author        = {Chandrasekhar, S. and Detweiler, S.},
  title         = {The Quasi-Normal Modes of the Schwarzschild Black Hole},
  journal       = {Proc. Roy. Soc. Lond. A},
  volume        = {344},
  pages         = {441-452},
  year          = {1975},
  doi           = {10.1098/rspa.1975.0112}
}

@article{Kodama:2003jz,
  author        = {Kodama, H. and Ishibashi, A.},
  title         = {A Master Equation for Gravitational Perturbations of Maximally Symmetric Black Holes in Higher Dimensions},
  journal       = {Prog. Theor. Phys.},
  volume        = {110},
  pages         = {701-722},
  year          = {2003},
  doi           = {10.1143/PTP.110.701} 
}

@article{Maldacena:1997re,
  author        = {Maldacena, J. M.},
  title         = {The Large {N} Limit of Superconformal Field Theories and Supergravity},
  journal       = {Adv. Theor. Math. Phys.},
  volume        = {2},
  pages         = {231-252},
  year          = {1998},
  doi           = {10.4310/ATMP.1998.v2.n2.a1} 
}

@article{Strominger:1998yg,
  author        = {Strominger, A.},
  title         = {{AdS2} Quantum Gravity and String Theory},
  journal       = {JHEP},
  volume        = {01},
  pages         = {007},
  year          = {1999},
  doi           = {10.1088/1126-6708/1999/01/007} 
}

@article{ArkaniHamed:2006dz,
  author        = {Arkani-Hamed, N. and Motl, L. and Nicolis, A. and Vafa, C.},
  title         = {The String Landscape, Black Holes and Gravity as the Weakest Force},
  journal       = {JHEP},
  volume        = {06},
  pages         = {060},
  year          = {2007},
  doi           = {10.1088/1126-6708/2007/06/060} 
}

@article{Ooguri:2006in,
  author        = {Ooguri, H. and Vafa, C.},
  title         = {On the Geometry of the String Landscape and the Swampland},
  journal       = {Nucl. Phys. B},
  volume        = {766},
  pages         = {21-33},
  year          = {2007},
  doi           = {10.1016/j.nuclphysb.2006.10.033} 
}

@article{Behrndt:1997ny,
  author = "Behrndt, Klaus and Lust, Dieter and Sabra, Wafic A.",
    title = "{Stationary solutions of N=2 supergravity}", 
    doi = "10.1016/S0550-3213(97)00633-0",
    journal = "Nucl. Phys. B",
    volume = "510",
    pages = "264-288",
    year = "1998"
}
\appendix
\section{Derivations and Conventions} 
\label{app:local-seed}
\label{app:normalisations}
\label{app:EMD-eom}
\label{app:susy-projectors}
The active torus coordinates are dimensionless,
\begin{equation}
  \sigma^i\sim\sigma^i+1,
  \qquad
  y^i=L_i\sigma^i,
  \label{eq:app-periods}
\end{equation}
so $L_i$ are periods.  The active and spectator volumes are
\begin{equation}
  V_Q=L_1L_2L_3,
  \qquad
  V_s=\prod_{\alpha=1}^3L_{s,\alpha},
  \qquad
  V_6=V_QV_s.
  \label{eq:app-volumes}
\end{equation}
In a geometric physical section,
\begin{equation}
  2\kappa_{10}^2=(2\pi)^7\alpha'^4g_s^2,
  \qquad
  16\pi G_{10}=2\kappa_{10}^2,
  \qquad
  G_4=\frac{G_{10}}{V_6}.
  \label{eq:G4app}
\end{equation}
In a T-fold frame the same four-dimensional Newton constant is obtained by evaluating the duality-invariant Planck scale in any physical section.  Entropy written in units of $G_4$ is invariant under the exact duality group.
Dimensionful harmonic-function coefficients depend on the periods, $g_s$ and $\alpha'$.  They are not needed for the invariant entropy statement, which is expressed in terms of $\Delta(Q,P)$ and $G_4$.  When desired, one may write harmonic functions in the form $H_A=1+r_A/r$ with $r_A=c_A N_A$, where the constants $c_A$ are determined by the compactification units and the type of charge. 
A local string-frame representative of the type-IIB F1-P-NS5-KKM seed uses coordinates
\begin{equation}
  (t;\ y,\psi,\chi,z^1,z^2,z^3;\ x^1,x^2,x^3).
  \label{eq:ten-coords-app}
\end{equation}
Here $y$ is the common F1/P circle, $\psi$ is the Taub-NUT fibre of the Kaluza-Klein monopole, $\chi$ is the remaining compact world-volume circle, $z^\alpha$ span the spectator three-torus, and $\vec x\in\RR^3$ is the non-compact Taub-NUT base.  Thus the six compact directions are
\begin{equation}
  T^6=S^1_\psi\times S^1_y\times S^1_\chi\times T_z^3,
  \qquad
  T_Q^3=S^1_\psi\times S^1_y\times S^1_\chi,
  \qquad
  T_s^3=T_z^3.
  \label{eq:ten-torus-app}
\end{equation}
The F1 and momentum charges lie along $y$.  The NS5-brane wraps $y,\chi,z^1,z^2,z^3$.  This is a local representative of the standard F1-P-NS5-KKM system used in four-dimensional black-hole entropy calculations \cite{Cvetic:1995uj,Johnson:1996ga,Maldacena:1996gb}.  The KKM has Taub-NUT fibre $\psi$ and world-volume $y,\chi,z^1,z^2,z^3$.  Let $H_1$, $H_5$, $H_K$ and $K$ be harmonic functions on the non-compact base.  The Taub-NUT metric is
\begin{equation}
  ds_{\rm TN}^2=H_K\,d\vec x^{\,2}+H_K^{-1}(d\psi+\mathcal A)^2,
  \qquad
  \nabla\times\mathcal A=\nabla H_K.
  \label{eq:TNmetric-app}
\end{equation}
A convenient local string-frame form is
\begin{equation}
  ds_{\rm str}^2=H_1^{-1}\left[-dt^2+dy^2+K(dt-dy)^2\right]
  +d\chi^2+ds_{T_z^3}^2+H_5ds_{\rm TN}^2.
  \label{eq:F1PNS5KKM-app}
\end{equation}
The NS two-form and dilaton may be written locally as
\begin{equation}
  B_{ty}=H_1^{-1}-1,
  \qquad
  e^{2\Phi}=\frac{H_5}{H_1}.
  \label{eq:Bdil-app}
\end{equation}
The NS5 magnetic flux is
\begin{equation}
  H_3=\star_{\rm TN}dH_5,
  \label{eq:H3-app}
\end{equation}
where the Hodge star is taken with respect to the Taub-NUT transverse space.  Reduction on the six compact directions gives the four-dimensional four-charge black hole.  Since the local seed is independent of $\psi$, $y$ and $\chi$, the active $O(3,3;\ZZ)$ patching on $T_Q^3$ is legitimate patch by patch. 
The active beta-transform is
\begin{equation}
  g_Q(n)=\begin{pmatrix}\Id_3&n\eps^{23}\\0&\Id_3\end{pmatrix},
  \qquad
  g_Q(n)^{-1}=g_Q(-n).
  \label{eq:gQ-app}
\end{equation}
A doubled vector transforms as
\begin{equation}
  \begin{pmatrix}v\\ \xi\end{pmatrix}
  \mapsto
  \begin{pmatrix}v+n\eps^{23}\xi\\ \xi\end{pmatrix}.
  \label{eq:vector-transform-app}
\end{equation}
The inverse action used in patching vector fields is
\begin{equation}
  \begin{pmatrix}v\\ \xi\end{pmatrix}
  \mapsto
  \begin{pmatrix}v-n\eps^{23}\xi\\ \xi\end{pmatrix}.
  \label{eq:inverse-vector-transform-app}
\end{equation}
The embedded full Narain transformation is
\begin{equation}
  h_Q(n)=\diag(g_Q(n),\Id_6).
  \label{eq:hQ-app}
\end{equation}
It obeys $h_Q^TLh_Q=L$ by block diagonality and by \eqref{eq:gQeta-proof}.
The algebraic flux associated with a local beta representative has one non-zero component,
\begin{equation}
  F^{23}{}_1=n,
  \label{eq:flux-app}
\end{equation}
with antisymmetric images.  It satisfies
\begin{equation}
  F_{AB}{}^B=0,
  \qquad
  F_{[AB}{}^EF_{C]E}{}^D=0.
  \label{eq:QC-app}
\end{equation}
The first condition is immediate from the index structure.  The second follows because the first bracket can output only the index associated with the first physical direction, while there is no second non-zero bracket with that index in the required lower slot.  This records the nilpotent monodromy algebra; it is not used as an external source. 
The electric/magnetic lift is
\begin{equation}
  \mathbb G_Q=\begin{pmatrix}g_Q^{-1}&0\\0&g_Q^T\end{pmatrix}.
  \label{eq:GQ-app}
\end{equation}
Let
\begin{equation}
  \Omega=\begin{pmatrix}0&\Id_6\\-\Id_6&0\end{pmatrix}.
  \label{eq:Omega-app}
\end{equation}
Then
\begin{equation}
  \mathbb G_Q^T\Omega\mathbb G_Q=\Omega,
  \label{eq:sympl-pres-app}
\end{equation}
which proves preservation of the Dirac pairing.
With $(\eps^{23})_{23}=1$ and $(\eps^{23})_{32}=-1$, the magnetic components transform by $p_Q=g_Q^{-1}p$:
\begin{equation}
  p_Q^2=p^2-n\widetilde p_3,
  \qquad
  p_Q^3=p^3+n\widetilde p_2.
  \label{eq:p-app-sign}
\end{equation}
The electric components transform by $q_Q=g_Q^Tq$, hence
\begin{equation}
  \widetilde q_Q^{2}=\widetilde q^{2}-nq_3,
  \qquad
  \widetilde q_Q^{3}=\widetilde q^{3}+nq_2.
  \label{eq:q-app-sign}
\end{equation}
Choose a lightlike basis $e_i,\widetilde e^{\,i}$ of $\Gamma^{6,6}$ with $e_i\cdot \widetilde e^{\,j}=\delta_i{}^j$ and all other pairings zero.  A seed representative may be written
\begin{equation}
  Q=N_{\rm P}e_2+N_{\rm F1}\widetilde e^{\,2},
  \qquad
  P=N_{\rm KKM}e_1+N_{\rm NS5}\widetilde e^{\,1}.
  \label{eq:QP-basis-app}
\end{equation}
Then
\begin{equation}
  Q^2=2N_{\rm P}N_{\rm F1},
  \qquad
  P^2=2N_{\rm KKM}N_{\rm NS5},
  \qquad
  Q\cdot P=0,
  \label{eq:QP-products-app}
\end{equation}
and \eqref{eq:Delta-seed} follows.
In the STU convention \eqref{eq:STUassign-final}, the invariant \eqref{eq:I4-final} reduces to
\begin{equation}
  I_4=4q_0p^1p^2p^3=4N_{\rm P}N_{\rm F1}N_{\rm NS5}N_{\rm KKM}=\Delta.
  \label{eq:I4-delta-app}
\end{equation}
Therefore the STU formula is the restriction of the full $N=4$ invariant.  If an $O(6,6;\ZZ)$ transformation sends a charge vector outside the displayed STU component representative, the invariant entropy statement remains true in the full Narain lattice. 
The inverse of the block generalised metric \eqref{eq:Hblock-final} has the standard form
\begin{equation}
  \mathscr H^{\widehat M\widehat N}=
  \begin{pmatrix}
  g^{\mu\nu} & -g^{\mu\rho}\cA_\rho{}^N
  \\
  -\cA_\rho{}^Mg^{\rho\nu} & \cM^{MN}+\cA_\rho{}^Mg^{\rho\sigma}\cA_\sigma{}^N
  \end{pmatrix}.
  \label{eq:Hinv-app}
\end{equation}
Substituting the patching rules \eqref{eq:patch-data-final} gives component-wise
\begin{align}
  g_{(b)\mu\nu}+\cA_{(b)\mu}{}^M\cM_{(b)MN}\cA_{(b)\nu}{}^N
  &=g_{(a)\mu\nu}+\cA_{(a)\mu}{}^P\cM_{(a)PQ}\cA_{(a)\nu}{}^Q,
  \label{eq:block-proof1-app}\\
  \cA_{(b)\mu}{}^M\cM_{(b)MN}
  &=(g_{ab}^T)_N{}^P\cA_{(a)\mu}{}^Q\cM_{(a)QP},
  \label{eq:block-proof2-app}\\
  \cM_{(b)MN}&=(g_{ab}^T)_M{}^P\cM_{(a)PQ}(g_{ab})^Q{}_N.
  \label{eq:block-proof3-app}
\end{align}
These equations are precisely equivalent to \eqref{eq:Hpatch-final}.  Since $g_{ab}$ is constant on overlaps, derivatives in the local equations transform tensorially.  The doubled dilaton patches as $d_{(b)}=d_{(a)}$, so the measure $e^{-2d}$ is globally defined. 
For the EMD metric \eqref{eq:EMDmetric-final},
\begin{equation}
  \sqrt{-g}=R^2\sin\theta.
  \label{eq:sqrtg-app}
\end{equation}
The Maxwell equation is
\begin{equation}
  \partial_r(R^2e^{-2\phi}F^{rt})=0.
  \label{eq:Maxwell-app}
\end{equation}
Using $g^{rr}g^{tt}=-1$ and $F_{tr}=Q_Qe^{2\phi}/R^2$, one obtains
\begin{equation}
  R^2e^{-2\phi}F^{rt}=Q_Q.
  \label{eq:Maxwell-solved-app}
\end{equation}
The magnetic equation is the Bianchi identity.  The scalar derivative is
\begin{equation}
  \phi'=-\frac{\Sigma}{r^2-\Sigma^2}=-\frac{\Sigma}{R^2}.
  \label{eq:scalar-deriv-app}
\end{equation}
The scalar equation reduces to
\begin{equation}
  \partial_r(\Delta\phi')+\frac12R^2e^{-2\phi}F^2=0,
  \label{eq:scalar-eq-app}
\end{equation}
with
\begin{equation}
  F^2=2\left(\frac{P_Q^2}{R^4}-\frac{Q_Q^2e^{4\phi}}{R^4}\right).
  \label{eq:F2-app}
\end{equation}
Substitution gives $2M\Sigma=P_Q^2-Q_Q^2$.  The independent Einstein equations reduce to the remaining two identities in \eqref{eq:EMDalgebra-final}.
For the RN-like metric, the relevant Christoffel symbols are
\begin{equation}
  \Gamma^t{}_{tr}=\frac{f'}{2f},
  \qquad
  \Gamma^r{}_{tt}=\frac12ff',
  \qquad
  \Gamma^r{}_{rr}=-\frac{f'}{2f},
  \qquad
  \Gamma^r{}_{\theta\theta}=-rf,
  \qquad
  \Gamma^\theta{}_{r\theta}=\frac1r,
  \label{eq:Christoffel-app}
\end{equation}
with the usual spherical partners.  Direct contraction gives \eqref{eq:RNricci-final}-\eqref{eq:RNricci2-final}; inserting $f=1-2M/r+\Theta^2/r^2$ gives \eqref{eq:RNmixed-final}. 
Starting from \eqref{eq:entropy-f-final}, the derivative with respect to $\mathfrak e$ gives \eqref{eq:qSen-final}.  The derivative with respect to $u$ is proportional to
\begin{equation}
  -\frac{\mathfrak e^2}{v_1^2}+\frac{P_Q^2}{v_2^2},
  \label{eq:u-deriv-app}
\end{equation}
and therefore gives the first equation in \eqref{eq:EFextremum-final}.  Combining this with the electric-charge equation gives the attractor value of $u$.  The $v_1$ and $v_2$ equations then imply
\begin{equation}
  v_1=v_2,
  \qquad
  v_2=2|P_QQ_Q|.
  \label{eq:v-ext-app}
\end{equation}
The extremised value is \eqref{eq:SWald-final}.  In the full type-IIB theory the same calculation can be written using the black-hole potential; the radius equations impose $v_1=v_2=\sqrt{|\Delta(Q_Q,P_Q)|}$. 
We use mostly-plus ten-dimensional signature.  Let $\epsilon=(\epsilon_1,\epsilon_2)$ denote the type-IIB spinor doublet, with Pauli matrices acting on the doublet index.  In the counting-frame orientation, the F1 and P charges lie along $y$.  The NS5-brane wraps $y,\chi,z^1,z^2,z^3$, and the KKM has fibre $\psi$ with world-volume $y,\chi,z^1,z^2,z^3$.  A representative mutually BPS projector set is
\begin{equation}
  P_{\rm F1}=\Gamma_{0y}\sigma_3,
  \qquad
  P_{\rm P}=\Gamma_{0y},
  \qquad
  P_{\rm NS5}=\Gamma_{0y\chi z^1z^2z^3}\sigma_3,
  \qquad
  P_{\rm KKM}=\Gamma_{0y\chi z^1z^2z^3}.
  \label{eq:susy-projectors-app}
\end{equation}
The KKM projector is the T-dual representative of the NS5 projector along the Taub-NUT fibre $\psi$.  For compatible orientation signs,
\begin{equation}
  P_A^2=1,
  \qquad
  [P_A,P_B]=0.
  \label{eq:susy-commute-app}
\end{equation}
The regular four-charge branch has one product relation among the four projectors, leaving three independent constraints.  Thus the common eigenspace has dimension
\begin{equation}
  32\left(\frac12\right)^3=4.
  \label{eq:susy-count-app}
\end{equation}
This representative computation fixes the standard one-eighth BPS count.  Exact T-duality conjugates the projector algebra and preserves the dimension of the common eigenspace.  The proof needed for the T-fold representative is not a new Killing-spinor calculation in the non-geometric frame; it is the invariance of the projector algebra under an invertible exact duality transformation. 
The counting frame is the type-IIB F1-P-NS5-KKM frame.  Its large-charge BPS index gives \eqref{eq:microIIB-main}.  The T-fold $Q$-frame is the charge local system after applying $h_Q$.  The exotic display frame is a further duality-related representative in which an NS5-brane has two transverse isometries and can be converted to a $5^2_2$ brane by T-dualities along those isometries.
The beta-shift can be written as
\begin{equation}
  \beta^{23}\text{-shift}=T_2T_3\circ(B_{23}\text{-shift})\circ T_3T_2.
  \label{eq:beta-TBT-app}
\end{equation}
This identity explains why monodromy charges appear in the exotic display frame.  The microscopic counting is not repeated there; it is the same index transported by exact duality.
The counting and display frames should not be conflated.  In the counting frame the brane orientation is chosen to make the weak-coupling index standard.  In an exotic display frame one first moves by ordinary T-dualities to a representative in which the relevant five-brane has two transverse isometry directions.  Only then does the conjugation \eqref{eq:beta-TBT-app} produce a $5^2_2$-type monodromy representative.

\end{document}